\documentclass[12pt]{article}%

\usepackage{amssymb}
\usepackage{amsfonts}
\usepackage{amsmath}
\usepackage[nohead]{geometry}
\usepackage[singlespacing]{setspace}
\usepackage[bottom]{footmisc}
\usepackage{indentfirst}
\usepackage{endnotes}
\usepackage{graphicx}%
\usepackage{rotating}
\usepackage{verbatim}
\usepackage{setspace}
\usepackage{multirow}
\usepackage{latexsym}
\usepackage{bbm}
\usepackage{amsfonts}

\usepackage{graphicx}
\usepackage{caption}
\usepackage{subcaption}

\usepackage{amsthm}
\usepackage{makeidx}
\usepackage{fancyhdr}
\usepackage{type1cm}

 \usepackage[authoryear]{natbib}

\newcommand{\bb}{\mathrm{\bf b}}
\newcommand{\bff}{\mathrm{\bf f}}
\newcommand{\bx}{\mathrm{\bf x}}
\newcommand{\by}{\mathrm{\bf y}}

\newcommand{\bt}{\mathrm{\bf t}}
\newcommand{\bA}{\mathrm{\bf A}}
\newcommand{\ba}{\mathrm{\bf a}}
\newcommand{\bu}{\mathrm{\bf u}}
\newcommand{\bv}{\mathrm{\bf v}}
\newcommand{\be}{\mathrm{\bf e}}
\newcommand{\bB}{\mathrm{\bf B}}
\newcommand{\bZ}{\mathrm{\bf Z}}

\newcommand{\bz}{\mathrm{\bf z}}
\newcommand{\bD}{\mathrm{\bf D}}

\newcommand{\bF}{\mathrm{\bf F}}
\newcommand{\bK}{\mathrm{\bf K}}
\newcommand{\bW}{\mathrm{\bf W}}
\newcommand{\bG}{\mathrm{\bf G}}
\newcommand{\bM}{\mathrm{\bf M}}
\newcommand{\bh}{\mathrm{\bf h}}
\newcommand{\bH}{\mathrm{\bf H}}
\newcommand{\bI}{\mathrm{\bf I}}
\newcommand{\bJ}{\mathrm{\bf J}}
\newcommand{\bL}{\mathrm{\bf L}}
\newcommand{\bg}{\mathrm{\bf g}}

\newcommand{\bV}{\mathrm{\bf V}}

\newcommand{\br}{\mathrm{\bf r}}
\newcommand{\bR}{\mathrm{\bf R}}

\newcommand{\bT}{\mathrm{\bf T}}
\newcommand{\bU}{\mathrm{\bf U}}
\newcommand{\bX}{\mathrm{\bf X}}
\newcommand{\bY}{\mathrm{\bf Y}}

\newcommand{\balpha}{\mbox{\boldmath $\alpha$}}

\newcommand{\bxi}{\mbox{\boldmath $\xi$}}

\newcommand{\bmu}{\mbox{\boldmath $\mu$}}
\newcommand{\bLambda}{\mbox{\boldmath $\Lambda$}}

\newcommand{\bgamma}{\mbox{\boldmath $\gamma$}}
\newcommand{\bGamma}{\mbox{\boldmath $\Gamma$}}

\newcommand{\bSigma}{\mbox{\boldmath $\Sigma$}}
\newcommand{\bOmega}{\mbox{\boldmath $\Omega$}}

\newcommand{\cov}{\mathrm{cov}}

\newcommand{\tr}{\mathrm{tr}}
\newcommand{\diag}{\mathrm{diag}}
\newcommand{\argmin}{\mathrm{argmin}}

\newcommand{\bw}{\mbox{\bf w}}

\def \diag {\mbox{diag}}

\def\rank{\mbox{rank}}
\def \bcalX {{\bf \mathcal{X}}}
\def \bcalL {{\bf \mathcal{L}}}

\def \calS {\mathcal{S}}
\newcommand{\bTheta}{\mbox{\boldmath $\Theta$}}

\newcommand{\beq}{\begin{equation}}
\newcommand{\eeq}{\end{equation}}
\newcommand{\beqn}{\begin{eqnarray}}
\newcommand{\eeqn}{\end{eqnarray}}
\newcommand{\beqnn}{\begin{eqnarray*}}
\newcommand{\eeqnn}{\end{eqnarray*}}

\DeclareMathOperator{\Cov}{Cov}
\DeclareMathOperator{\sgn}{sgn}
\DeclareMathOperator{\card}{card}

\numberwithin{equation}{section}
\theoremstyle{plain}
\newtheorem{thm}{Theorem}[section]

\newtheorem{lem}{Lemma}[section]

\newtheorem{assum}{Assumption}[section]
\theoremstyle{definition}

\makeatletter
\def\@biblabel#1{\hspace*{-\labelsep}}
\makeatother
\begin{document}

 \title{Large Covariance Estimation through Elliptical Factor Models}
\author{Jianqing Fan\thanks{The research was partially supported by NSF grants
DMS-1206464 and DMS-1406266 and NIH grants R01-GM072611-10
and NIH R01GM100474-04.}, \, Han Liu\thanks{The research was supported by NSF CAREER Award DMS1454377, NSF IIS1408910, NSF IIS1332109, NIH R01MH102339, NIH R01GM083084, and NIH R01HG06841.} \, and Weichen Wang\thanks{Address: Department of ORFE, Sherrerd Hall, Princeton University, Princeton, NJ 08544, USA, e-mail: \textit{jqfan@princeton.edu}, \textit{hanliu@princeton.edu}, \textit{weichenw@princeton.edu}. }
\medskip\\{\normalsize Department of Operations Research and Financial Engineering,  Princeton University}}
 
\date{}

\maketitle

\sloppy

\onehalfspacing
 
\begin{abstract}
We proposed  a general Principal Orthogonal complEment Thresholding (POET) framework for large-scale covariance matrix estimation based on an approximate factor model. A set of high level sufficient conditions for the procedure to achieve optimal rates of convergence under different matrix norms were brought up to better understand how POET works. Such a framework allows us to recover the results for sub-Gaussian in a more transparent way that only depends on the concentration properties of the sample covariance matrix. As a new theoretical contribution,  for the first time, such a framework allows us to exploit conditional sparsity covariance structure for the heavy-tailed data.  In particular, for  the elliptical data, we proposed a robust estimator based on marginal and multivariate Kendall's tau to satisfy these conditions. In addition, conditional graphical model was also studied under the same framework. The technical tools developed in this paper are of general interest to high dimensional principal component analysis. Thorough numerical results were also provided to back up the developed theory.
\end{abstract}


\pagebreak%
\doublespacing

\onehalfspacing

\section{Introduction}

This paper considers large factor model based covariance matrix estimation for heavy-tailed data. Factor model is a powerful tool for dimension reduction and latent factor extraction, which gained its popularity in various applications from finance to biology. When applied to covariance matrix estimation, it assumes a conditional sparse covariance structure, i.e., conditioning on the low dimensional spiked factors, the covariance matrix of the idiosyncratic errors is sparse. To be specific, consider the approximate factor model in \cite{BaiNg02}:
\beq \label{eq:factor}
y_{it} = \bb_i' \bff_t + u_{it} \,,
\eeq
where $y_{it}$ is the observed data for the $i$th ($i=1,\dots,p$) dimension at time $t=1,\dots,n$; $\bff_t$ is an unknown $m$-dimensional  vector of common factors, and $\bb_i$ is the factor loading for the $i$th variable; $u_{it}$ is the idiosyncratic error, uncorrelated with the common factors. Previous works are limited by only considering Gaussian or sub-Gaussian factors and noises.  In this paper we aim to extend this limitation and consider heavy-tailed distributions. More specifically, we will consider the case where factors and noises are elliptically distributed. Under this broader class of heavy tailed distributions, we aim to understand how to estimate covariance matrix accurately.

Covariance matrix estimation has been pioneered by  \cite{BicLev08a,BicLev08b} and \cite{FanFanLv08}. After that, substantial amount of work has focused on the inference of high-dimensional covariance matrices under unconditional sparsity \citep{CaiLiu11,CaiRenZho13,
CaiZhaZho10,
Kar08, LamFan09, RavWaiRasYu11} or conditional sparsity \citep{AmiWai09,BerRig13a, BerRig13b, BirJohNadPau13, CaiMaWu13, CaiMaWu14,  JohLu09, LevVer12,RotLevZhu09, Ma13, SheSheMar11, PauJoh12,VuLei12, ZouHasTib06}. This research  area  is very active, and as a result, this list of references is illustrative rather than comprehensive. To emphasize, Fan and his collaborators proposed to use factor model or conditional sparsity structure for covariance matrix estimation \citep{FanFanLv08,FanLiaMin11,FanLiaMin13, FanLiaWan14}. The model encompasses the situation of  unconditional sparse covariance by setting the number of factors to zero. Thus it is more general and  realistic given the fact that the observed data are usually driven by some common factors.

Another line of research on robust covariance estimation also receives significant attention from the literature. The idea of robust estimation dates back to \cite{Hub64} and had been extended in regression problems with different types of loss function; see for example \cite{FanLiWan14} and \cite{Cat12}. Recently, \cite{HanLiu13a, HanLiu14} introduce robust covariance matrix estimation to high-dimensional elliptical and transelliptical (or elliptical copula) distribution family. In those papers, they proposed a robust procedure using the marginal Kendall's tau statistics and proved its optimality for covariance matrix estimation under elliptical distributions. In addition, multivariate Kendall's tau was also considered by \cite{HanLiu13b} to estimate eigenspaces of covariance matrices in  high dimensions. Those methods, applied to PCA or sparse PCA, can be potentially useful for dealing with factor models with heavy-tailed factors and noises. The goal of the current paper is
to develop a unified theory that allows us to extend these robust rank-based covariance estimation procedures to handle heavy-tailed data with conditional covariance sparsity.

\subsection{Background on approximate factor model}

To illustrate how to use factor model as a dimension reduction tool for covariance matrix estimation, let us write model (\ref{eq:factor}) in its vector form:
\beq\label{eq::model1}
\by_{t} = \bB \bff_t + \bu_{t}\,,
\eeq
where $\by_t$ contains all observed individuals at time $t = 1,\dots, n$ and $\bB = (\bb_1, \dots, \bb_p)'$ is the factor loading matrix. The matrix form of (\ref{eq:factor}) is
\beq
\bY = \bB \bF' +\bU\,,
\eeq
where $\bY_{p\times n}$, $\bB_{p\times m}$, $\bF_{n \times m}$, $\bU_{p\times n}$ are  matrices from  observed data, factor loadings, factors, and errors with $\bY = (\by_1, \dots, \by_n)$, $\bF = (\bff_1, \dots, \bff_n)'$ and $\bU = (\bu_1, \dots, \bu_n)$. Here we consider the case where the dimension $p$ is larger than sample size $n$ and for simplicity we assume $n$ samples are independent and identically distributed in the sequel (An extension to the dependent setting is straightforward, but tedious.). We assume factor matrix $\bF$ is observable. To make the model (\ref{eq:factor}) identifiable, we impose the following  conditions  as in \cite{BaiNg13} and \cite{BaiLi12}:
\beq\label{eq::identifibility}
\cov(\bff_t) = \bI \text{ and } \bB'\bB \text{ is diagonal}\,.
\eeq
The conditions in \eqref{eq::identifibility} are common in the factor model literature. But we will point out in Section \ref{sec2} that these conditions are sufficient only for asymptotic identifiability up to an error of order $O(1/\sqrt{p})$ rather than exact identifiability.
Under the conditions in \eqref{eq::identifibility}, the covariance matrix of $\by_t$ is
\beq \label{eq:matrixDecomp}
\bSigma = \cov(\by_t) = \bB\bB' + \bSigma_u\,,
\eeq
where $\bSigma_u$ is the covariance matrix of the idiosyncratic error $\bu_t$.

\subsection{Major contributions of this paper}

Under  model \eqref{eq::model1},  \cite{FanLiaMin13} proposed the Principal Orthogonal complEment Thresholding (POET) estimator for $\bSigma$ under the assumption that  factors and noises are sub-Gaussian.
By imposing the condition that the leading eigenvalues of $\bSigma$ diverges at the rate of order $p$ from their pervasiveness condition, \cite{FanLiaMin13} proved the consistency of the POET estimator and showed its rates of convergence.
However, their proofs are mathematically involved and do not transparently explain why POET works in estimating large covariance matrices.
It has been pointed out by \cite{FanWan15} how pervasive factors help in estimating the low-rank part $\bB\bB'$ in (\ref{eq:matrixDecomp}). The idea is further explored in this paper. A surprising result is that the diverging signal of spiked eigenvalues excludes the necessity of the sparse principal component assumption in sparse PCA literature, comparing with for example \cite{CaiMaWu13}.

The main contributions of the paper are two folds. On one hand, we summarize a unified generic framework in Section \ref{sec2.1} for applying POET to various potentially heavy-tailed distributions. The key Theorem \ref{suff} provides a set of high level interface conditions (\ref{suffCond}) explaining how to design a POET covariance estimator according to factor and error distributions. POET regularization needs the following three components: initial pilot estimators for covariance matrix $\bSigma$, its leading eigenvalues $\bLambda = \diag(\lambda_1, \dots, \lambda_m)$ and their corresponding leading eigenvectors $\bGamma_{p \times m} = (\bxi_1, \dots, \bxi_m)$.  With these compoents, a generic POET estimator can be constructed.
We will show that such a POET procedure attains desired rates of convergence as long as
\begin{equation} \label{suffCond}
\begin{aligned}
&\|\hat\bSigma - \bSigma\|_{\max} = O_P(\sqrt{\log p/n})\,, \\
&\|(\hat\bLambda - \bLambda)\bLambda^{-1}\|_{\max} = O_P(\sqrt{\log p/n})\,, \\
&\|\hat\bGamma - \bGamma\|_{\max} = O_P(\sqrt{\log p/(np)}) \,.
\end{aligned}
\end{equation}
These conditions are relatively easy to verify, as they involve only the componentwise maximums.
Through those sufficient conditions, we are able to separate the deterministic analysis of the estimation procedure and the probabilistic guarantee of the design of initial estimators.

For two specific factor and error distributions, we provide methods to construct those initial estimators. For sub-Gaussian, it is natural to employ sample covariance matrix and its eigenvalues and eigenvectors as the estimates for $\bSigma$, $\bLambda$ and $\bGamma$. We show the natural idea indeed achieves the above conditions for sub-Gaussian data, which gives an explanation why POET in previous literature works. However, for elliptical distributions, constructing estimators with the desired rates are highly nontrivial. We use the marginal Kendall's tau to obtain $\hat\bSigma$ and $\hat\bLambda$ while a different method multivariate Kendall's tau is applied to construct $\hat\bGamma$. Notice an interesting fact that the generic POET procedure allows separately estimating the eigenvectors and eigenvalues using different methods. Robust estimators are constructed for the first time for elliptical factor models.

\subsection{Notations} Here are some useful notations. If $\bM$ is a general matrix, we denote its matrix entry-wise maximum value as $\|\bM\|_{\max} = \max_{i,j}|M_{i,j}|$ and define the quantities $\|\bM\|_2 = \lambda_{\max}^{1/2}(\bM'\bM)$ (or $\|\bM\|$ for short), $\|\bM\|_F = (\sum_{i,j} M_{i,j}^2)^{1/2}$, $\|\bM\|_{\infty} = \max_{i} \sum_j |M_{i,j}|$ and $\|\bM\|_{1,1} = \sum_{i} \sum_j |M_{i,j}|$ to be its spectral, Frobenius, induced $\ell_{\infty}$ and element-wise $\ell_1$ norms. If furthermore $\bM$ is symmetric, we define $\lambda_{j}(\bM)$ to be the $j$th largest eigenvalue of $\bM$ and $\lambda_{\max}(\bM)$, $\lambda_{\min}(\bM)$ to be the maximal and minimal eigenvalues respectively. We denote $\tr(\bM)$ to be the trace of $\bM$. For any vector $\bv$, its $\ell_2$ norm is represented by $\|\bv\|$ while $\ell_1$ norm is written as $\|\bv\|_1$. We denote $\diag(\bv)$ to be the diagonal matrix with the same diagonal entries as $\bv$. For two random matrices $\bA, \bB$ of the same size, we say $\bA = \bB + O_P(\delta)$ if $\|\bA- \bB\| = O_P(\delta)$ and $\bA = \bB + o_P(\delta)$ if $\|\bA - \bB\| = o_P(\delta)$. Similarly for two random vectors $\ba, \bb$ of the same length, $\ba = \bb + O_P(\delta)$ if $\|\ba - \bb\| = O_P(\delta)$ and $\ba = \bb + o_P(\delta)$ if $\|\ba - \bb\| = o_P(\delta)$.
We denote $\ba \overset{d} = \bb$ if random vectors $\ba$ and $\bb$ have the same distribution. In the sequel, $C$ is a generic constant that may differ from line to line.

\subsection{Paper organization} In Section \ref{sec2}, we present a generic POET estimating procedure and a  high-level theoretical interface which secures the consistency of the generic procedure for factor-based conditional sparsity mdoels. We verify that the conditions in Section \ref{sec3} hold with high probability for sub-Gaussian data, which provides a transparent understanding of the mechanism of the POET methodology. In Section \ref{sec4}, we propose a new method using a combination of marginal and multivariate Kendall's tau  and prove its theoretical properties under elliptical factor models. Thorough numerical simulations are conducted illustrate the merits of our proposed method in Section \ref{sec5}. In Section \ref{sec6},  we conclude the paper with a short discussion. The technical proofs are relegated to the appendix.

\section{A High-level theoretical interface} \label{sec2}
In this section, we summarize a generic POET procedure and provide a set of high level sufficient conditions for consistent covariance estimation when $p\gg n$. Before doing that, let us review what has been achieved in the existing literature where both the factors and noises are assumed to be sub-Gaussian.

\subsection{Spiked covariance model}

Assume the observed random variables $\{\by_i\}_{i=1}^n$ have zero mean and covariance matrix $\bSigma_{p\times p}$ where the eigenvalues $\lambda_1, \lambda_2, \dots, \lambda_p$ of $\bSigma$ are ordered in descending order. We consider the spiked population model as suggested by the approximate factor structure (\ref{eq:matrixDecomp}). Specifically we have the following assumption on the eigvenvalues.
\begin{assum}[Spiked covariance model]
\label{assump1}
Let $m\leq \min\{n,p\}$ be a fixed constant that does not change with $n$ and $p$.  As $n \to \infty$, $\lambda_1 > \lambda_2 > \dots > \lambda_m \gg \lambda_{m+1} \ge \dots \ge \lambda_p > 0$, where the spiked eigenvalues are linearly proportional to dimension $p$ while the non-spiked eigenvalues are bounded, i.e., $c_0 \le \lambda_j \le C_0, j>m$ for constants $c_0, C_0 > 0$. In addition, the non-spiked eigenvalue average $(p-m)^{-1} \sum_{j = m+1}^p \lambda_j = \bar c + o(1)$.
\end{assum}

Assumption \ref{assump1} requires the eigenvalues be divided into the diverging and  bounded ones.    For simplicity,  we only consider distinguishable eigenvalues (multiplicity 1) for the largest $m$ eigenvalues.  This assumption is typically satisfied by the factor model (\ref{eq:factor}) with pervasive factors. More specifically, if the factor loadings $\{\bb_j\}_{j=1}^p$ (the transpose of the rows of $\bB$) are an i.i.d. sample from a population with finite second moments, then by the strong law of large numbers, $ p^{-1} \bB' \bB = p^{-1} \sum_{j=1}^p \bb_j \bb_j' \to \bSigma_b$ almost surely, where  $\bSigma_b =  \mathbb{E} (\bb_j \bb_j')$.  In other words, the eigenvalues of $\bB \bB'$ are approximately
\[
  p \lambda_1(\bSigma_b) (1+o(1)), \cdots,  p \lambda_m(\bSigma_b) (1+o(1)), 0, \cdots, 0,
\]
where $\lambda_j(\bSigma_b)$ is the $j$th eigenvalue of $\bSigma_b$.
If we further assume that $\|\bSigma_u\|$ is bounded, by Weyl's theorem, we conclude
\begin{equation} \label{eq2.1}
  \lambda_j = p \lambda_j(\bSigma_b) (1+o(1)),  \quad \mbox{for } j = 1, \cdots, m,
\end{equation}
and the remaining are bounded.

\subsection{A generic POET procedure for covariance estimation} \label{sec2.1}

We see from (\ref{eq:matrixDecomp}) that the population covariance of the factor model (\ref{eq:factor}) exhibits a low-rank plus sparse structure, if 
$\bSigma_u$ is sparse, whose sparsity level is measured by
\[
m_p := \max_{i \le p} \sum_{j \le p} |\sigma_{u,ij}|^q
\]
for some $q \in [0,1]$ is small. In particular, with $q=0$, $m_p$ corresponds to the maximum number of nonzero elements in each row of $\bSigma_u$.

To estimate the covariance matrix $\bSigma$ with the approximate factor structure (\ref{eq:matrixDecomp}), \cite{FanLiaMin13} proposed the POET method to recover the factor matrix as well as the factor loadings. The idea is to first decompose the sample covariance matrix into the spike and non-spike parts,
\beq
\label{Eqn:decomp}
\hat\bSigma = \frac1n \sum_{i=1}^n \by_i \by_i' = \sum_{j = 1}^{m} \hat\lambda_j \hat\bxi_{j} \hat\bxi_j' + \hat\bSigma_u\,,
\eeq
where $\hat\bSigma_u = \sum_{j = m+1}^{p} \hat\lambda_j \hat\bxi_{j} \hat\bxi_j'$ is called the principal orthogonal complement. Then by employing adaptive thresholding on $\hat\bSigma_u$ to get $\hat\bSigma_u^{\top}$ \citep{CaiLiu11}, they obtain a final covariance estimator $\hat\bSigma^{\top}$ defined as
\beq
\hat\bSigma^{\top} = \sum_{j = 1}^{m} \hat\lambda_j \hat\bxi_{j} \hat\bxi_j' + \hat\bSigma_u^{\top} \,.
\eeq
The above procedure can be equivalently viewed as a least-squares approach. That is, the factor and loading matrices can be estimated by solving the following nonconvex minimization problem:
\beq \label{minimization}
(\hat\bB, \hat\bF) = \arg \min_{\bB,\bF} \|\bY - \bB\bF'\|_F^2 \text{ s.t. } \frac{1}{n} \bF'\bF = \bI_m, \bB'\bB \text{ is diagonal}.
\eeq
It is shown that the columns of $\hat\bF/\sqrt{n}$ are the eigenvectors corresponding to the $m$ largest eigenvalues of the $n\times n$ matrix $n^{-1} \bY'\bY$ and $\hat\bB = n^{-1} \bY \hat\bF$. Note that the estimator $\hat\bB$ given by minimizing (\ref{minimization}), after normalization, is actually the first $m$ empirical eigenvectors of $n^{-1} \bY \bY'$. Given $\hat\bB,\hat\bF$, we define $\hat \bU = \bY - \hat\bB \hat\bF'$ and $\hat\bSigma_u = n^{-1} \hat\bU \hat\bU'$. Finally adaptive thresholding  is applied to $\hat\bSigma_u$ to obtain $\hat\bSigma_u^{\top} = (\hat\sigma_{u,ij}^{\top})_{p\times p}$ with
\begin{equation} \label{thresholding}
\hat\sigma_{u,ij}^{\top} = \left\{  \begin{array}{lr} \hat\sigma_{u,ij}, & i = j\\
s_{ij} (\hat\sigma_{u,ij}) I(|\hat\sigma_{u,ij}| \ge \tau_{ij}),  & i \ne j\end{array} \right.,
\end{equation}
where $s_{ij}(\cdot)$ is the generalized shrinkage function \citep{AntFan01,RotLevZhu09} and $\tau_{ij} = \tau (\hat\sigma_{u,ii} \hat\sigma_{u,jj})^{1/2}$ is an entry-dependent threshold. The above adaptive threshold operator corresponds to applying thresholding with parameter $\tau$ to the correlation matrix of $\hat\bSigma_u$. The positive parameter $\tau$ will be determined based on theoretical analysis.

Let $w_n = \sqrt{\log p/n} + 1/\sqrt{p}$. \cite{FanLiaMin13} claimed that under some technical assumptions, with $\tau \asymp w_n$, if $m_p w_n^{1-q} = o(1)$,
\beq \label{rate1}
\|\hat\bSigma_u^{\top} - \bSigma_u\|_{2}  = O_P\Big( m_p w_n^{1-q} \Big) = \|(\hat\bSigma_u^{\top})^{-1} - {\bSigma_u}^{-1}\|_{2}\,,
\eeq
and
\beq \label{rate2}
\begin{aligned}
&\|\hat\bSigma^{\top} - \bSigma\|_{\max} = O_P\Big( w_n \Big)\,, \\
&\|\hat\bSigma^{\top} - \bSigma\|_{\bSigma} =  O_P\Big( \frac{\sqrt{p} \log p}{n} + m_p w_n^{1-q} \Big)\,, \\
&\|(\hat\bSigma^{\top})^{-1} - \bSigma^{-1}\|_2 = O_P\Big( m_p w_n^{1-q} \Big)\,,
\end{aligned}
\eeq
where $\|\bA\|_{\bSigma} = p^{-1/2} \|\bSigma^{-1/2}\bA \bSigma^{-1/2}\|_F$ is the relative Frobenius norm. The scaling $p^{-1/2}$ is exploited to ensure $\|\bSigma\|_{\bSigma} = 1$. The term $1/\sqrt{p}$ in $w_n$ is the price we need to pay for estimating the unknown factors. But in the high dimensional regime $p \ge n$ so that $1/\sqrt{p} \leq \sqrt{\log p/n}$, the rate is optimal. The original proofs for getting the above rates are mathematically involved and is not clear why the optimal rates can be attained, especially when no sparsity assumption for eigenvectors was imposed as in sparse PCA literature.

We propose a generic POET procedure here: (1) given three initial pilot estimators $\hat\bSigma, \hat\bLambda, \hat\bGamma$ for true covariance matrix $\bSigma$, leading eigenvalues $\bLambda = \diag(\lambda_1, \dots, \lambda_m)$ and leading eigenvectors $\bGamma_{p \times m} = (\bxi_1, \dots, \bxi_m)$ respectively, the principal orthogonal complement $\hat\bSigma_u$ can be computed by subtracting out the leading low-rank part, i.e.,
$$
    \hat\bSigma_u = \hat\bSigma -  \hat\bGamma \hat\bLambda \hat\bGamma';
$$
(2) The adaptive thresholding (\ref{thresholding}) is applied to $\hat\bSigma_u$ to obtain $\hat\bSigma_u^{\top}$, and (3) the low-rank structure is added back to obtain $\hat\bSigma^{\top}$. Note for sub-Gaussian distributions, $\hat\bLambda = \diag(\hat\lambda_1, \dots,\hat\lambda_m)$ is the diagonal matrix constructed by the first $m$ leading empirical eigenvalues of the sample covariance matrix $\hat\bSigma$ while $\hat\bGamma = (\hat\bxi_1, \dots, \hat\bxi_m)$ is the matrix of corresponding leading empirical eigenvectors. But in general, $\hat\bLambda$ and $\hat\bGamma$ do not have to come from the sample covariance matrix. In fact, they can even be separately estimated.

So our question is: why such a simple  POET procedure works under the piked covariance assumption \eqref{assump1}? Can we replace the sample covariance matrix by other pilot estimators as a starting point for the eigen-strucuture if other family of distributions, such as elliptical distributions or other more general heavy-tailed distributions, are considered?

\subsection{A high level theoretical interface}

A high level explanation is provided to understand the generic POET procedure. Sufficient conditions are brought up for $\hat\bSigma_u^{\top}$ and $\hat\bSigma^{\top}$ to achieve the desired rates of convergence  in (\ref{rate1}) and (\ref{rate2}). Our vital conclusion is stated in the following theorem.

\begin{thm} \label{suff}
Under Assumptions \ref{assump1}, if $\exists C > 0$ such that $\|\bB\|_{\max} \le C$ and $C^{-1} \le \|\bSigma_u \|_2 \le C$. If we have estimators $\hat\bSigma, \hat\bGamma, \hat\bLambda$ satisfying (\ref{suffCond}), then the rates of convergence in (\ref{rate1}) and (\ref{rate2}) hold with the generic POET procedure described in Section \ref{sec2.1}.
\end{thm}

The proof given in the appendix to obtain (\ref{goal}) provides insights on how the generic POET procedure works. Note that the max norm of low rank matrix estimation is bounded by $\Delta_1$ and $\Delta_2$. The former quantifies the estimation error of leading empirical eigen-structure $\hat\bGamma\hat\bLambda\hat\bGamma$ for its population counterpart, while the latter measures the error of identifying the low rank matrix $\bB \bB'$ by $\bGamma\bLambda\bGamma'$ from the true matrix $\bSigma$. The identification of low-rank and sparse matrices under pervasive condition is asymptotically unique with identification error $\Delta_2 = O(1/\sqrt{p})$. Additionally, the estimation contributes an error term of order $\Delta_1 = O_P(\sqrt{\log p/n})$.


\subsection{Conditional graphical model} \label{sec2.3}
In Section \ref{sec2.1}, $m_p$ measures the sparsity of $\bSigma_u$, but its inverse $\bOmega_u = \bSigma_u^{-1}$ is not necessarily sparse. Sometimes, the sparsity structure on $\bOmega_u$  reveals more interesting structure than $\bSigma_u$. For example,  If $\bu_t \sim EC_p({\bf 0}, \bSigma_u, \zeta)$, the sparsity of $\bOmega_u$ encodes the conditional uncorrelatedness relationships between all variables in the $p$ dimensional vector $\bu_t$. More specifically, for $p$ nodes $u_1, \dots, u_p$, each corresponding to one element of $\bu_t$, $u_i$ and $u_j$ are connected if and only if $(\bOmega_u)_{ij} \ne 0$, meaning that $u_{it}$ and $u_{jt}$ are uncorrelated conditioning on all the other $\{u_{kt}\}_{k \ne i, j}$ and $\bff_t$. If the number of factors is zero, this reduces to the classical elliptical graphical model, exhaustively studied by  \cite{VogFri11} and \cite{LiuHanZha12}.

In many applications, the conditional graphical model (or conditional sparse inverse covariance model) appears more natural compared to the conditional sparse covariance model. For example, in understanding the dependence of financial returns, the interest lies in the graphical model of the idiosyncratic components after taking the common market risk factors away; in genomic studies, the graphs after taking the confounding factors such as age and environment exposure are of better interest. The factors can be interpreted as covariates that need to be adjusted before focusing on the analysis of correlatedness of the residual part \citep{FanLiaMin11}. \cite{CaiLiLiuXie12} adopted the same idea of adjusting the factors in genomics application, but they do not assume the factors are pervasive so that they need to impose the constraint of a sparse factor loading matrix $\bB$. The sparsity was put on $\bOmega_u$ and measured by the quantity
\[
        M_p : = \max_{i \le p} \sum_{j \le p} |\omega_{u,ij}|^q.
\]

The generic POET procedure could also be modified to estimate conditional graphical model. The first step is still recovering $\hat \bSigma_u = \hat\bSigma - \hat\bGamma \hat\bLambda\hat\bGamma'$ by removing the effect of low-rank dominating factors. Then the method ``constrained
$\ell_1$-minimization for inverse matrix estimation'' (CLIME) proposed by \cite{CaiLiuLuo11} can be applied to obtain $\hat\bOmega_u$. Specifically, CLIME solves the following constrained minimization problem:
\beq \label{CLIME}
\hat\bOmega_u^1 = \argmin_{\bOmega} \|\bOmega\|_{1,1} \;\; \text{subject to} \;\; \|\hat\bSigma_u \bOmega - \bI\|_{\max} \le \tau,
\eeq
where $\|\bOmega\|_{1,1} = \sum_{i} \sum_j |\omega_{i,j}|$ and $\tau$ is a tuning parameter so that $\tau \asymp w_n$. A further symmetrization step can be carried out to guarantee a symmetric estimator $\hat\bOmega_u = (\hat\omega_{u, ij})$ where
\beq
\hat\omega_{u, ij} = \hat\omega_{u, ij}^1 \mathbbm{1}(|\hat\omega_{u, ij}^1| \le |\hat\omega_{u, ji}^1|) +  \hat\omega_{u, ji}^1 \mathbbm{1}(|\hat\omega_{u, ij}^1| > |\hat\omega_{u, ji}^1|)\,.
\eeq

Note that the optimization in (\ref{CLIME}) can be solved column by column using linear programming. Other possible methods can also be considered including graphical Lasso, graphical SCAD, graphical Dantzig selector, and graphical neighborhood selection \citep{FriHasTib08, YuaLin07, FanFenWu09,  LamFan09, RavWaiRasYu11, Yua10, MeiBuh06}.  Though substantial amount of efforts have been made to understand the graphical model,  little has been done for estimating conditional graphical model, which is again more general and realistic.

Once we have $\hat\bOmega_u$, the original inverse covariance matrix $\bOmega = \bSigma^{-1}$ can also be estimated using the Sherman-Morrison-Woodbury formula as follows:
\beq
\hat\bOmega = \hat\bOmega_u -  \hat\bOmega_u \hat\bGamma (\hat\bLambda^{-1} + \hat\bGamma' \hat\bOmega_u \hat\bGamma)^{-1} \hat\bGamma' \hat\bOmega_u\,.
\eeq
The following theorem gives the rates of convergence for $\hat\bOmega_u$ and $\hat\bOmega$ provided good pilot estimators $\hat\bSigma$, $\hat\bLambda$ and $\hat\bGamma$ are given. Its proof is in Appendix \ref{secA}.

\begin{thm} \label{Graph}
Under Assumptions \ref{assump1}, if $\exists C > 0$ such that $\|\bB\|_{\max} \le C$ and $C^{-1} \le \|\bOmega_u \|_2  \le \|\bOmega_u \|_{\infty} \le C$ and the estimators $\hat\bSigma, \hat\bGamma, \hat\bLambda$ satisfy conditions (\ref{suffCond}). Then the generic POET procedure with CLIME gives
\beq \label{rate3}
\begin{aligned}
&\|\hat\bOmega_u - \bOmega_u\|_{\max} = O_P( w_n)  = \|\hat\bOmega - \bOmega\|_{\max}\,; \\
& \|\hat\bOmega_u - \bOmega_u\|_{2} = O_P( M_p w_n^{1-q}) = \|\hat\bOmega - \bOmega\|_2 \,.
\end{aligned}
\eeq
\end{thm}
Note  the assumption of bounded $\|\bOmega_u \|_{\infty}$ is stronger than the case of estimating covariance matrix. This condition might be relaxed if other methods instead of CLIME was applied. But we do not pursue the weakest possible conditions here. Many potential applications are only involved with the estimation of inverse covariance matrix $\bOmega$, for instance classification and discriminant analyses and optimal portfolio allocation in finance.

\subsection{Positive semi-definite projection under max norm}

There is an additional issue that requires careful  consideration. In the generic POET procedure, if $\hat\bGamma$ and $\hat\bLambda$ are not estimated from the same positive semi-definite (PSD) matrix $\hat\bSigma$, the residual $\hat\bSigma_u$ may not be PSD for a given sample.
Thus, the following optimization should be considered to find the nearest PSD matrix of $\hat\bSigma_u$ in terms of the max norm:
\beq\label{projection}
\widetilde \bSigma_u = \argmin_{\bSigma_u \succeq {\bf 0}} \|\hat\bSigma_u - \bSigma_u\|_{\max}\,.
\eeq
The minimizer  preserves the max norm error bound since
\[
\|\widetilde\bSigma_u - \bSigma_u\|_{\max} \le \|\widetilde\bSigma_u - \hat\bSigma_u\|_{\max} + \|\hat\bSigma_u - \bSigma_u\|_{\max} \le 2 \|\hat\bSigma_u - \bSigma_u\|_{\max}\,,
\]
and everything else in the POET procedure works with $\hat\bSigma_u$ replaced by $\widetilde \bSigma_u$.
The same problem occurs in conditional graphical model estimation. Although $\hat\bOmega_u$ is  PSD with high probability, in practice we may reach a non-PSD estimator for $\bOmega_u$. So we need to explicitly perform the PSD projection of $\hat\bOmega_u$ onto the PSD cone as in (\ref{projection}).

Minimization (\ref{projection}) is challenging due to its non-smoothness. An effective smooth surrogate for the max norm objective was proposed by \cite{ZhaRoeLiu14} which can be solved efficiently. Specifically, they considered minimizing $\|\hat\bSigma_u - \bSigma_u\|_{\max}^{\mu}$ subject to $\bSigma_u \succeq {\bf 0}$ where
\[
\|\bA\|_{\max}^{\mu} = \max_{\|\bU\|_{1,1} \le 1} \langle \bU, \bA \rangle - \frac{\mu}{2}\|\bU\|_F^2\,,
\]
where $\|\bU\|_{1,1} = \sum_{i,j} |u_{ij}|$. More details can be found in  \cite{ZhaRoeLiu14}.
 Another possibility to ease computation burden is solving the dual problem of graphical lasso, that is,
\[
\max_{\bW} \log \det(\bW) \;\; \text{subject to} \;\; \|\bW - \hat\bSigma\|_{\max} \le \tau\,.
\]
By choosing $\tau \asymp w_n$, the optimal solution is a PSD matrix satisfying the max norm bound. Such a projection is still valid for the generic POET procedure to get the desired convergence rates under max norm.

\section{Sub-Gaussian factor models} \label{sec3}

We have established sufficient conditions in (\ref{suffCond}) for optimal estimation of covariance matrices as well as conditional graphical models. The next natural question is whether these conditions hold for sub-Gaussian factor models.  In this subsection, we validate the conditions for sample covariance matrix under sub-Gaussian conditions.

By the spectral decomposition, $\bSigma = \bGamma_p \bLambda_p \bGamma_p'$ where $\bLambda_p = \diag(\lambda _1, \dots,\lambda _p)$ and $\bGamma_p$ is constructed by all the corresponding eigenvectors of $\bSigma$. We use subscript $p$ to explicitly denote  the dependence of $\bLambda_p$ and $\bGamma_p$ on all eigenvalues or eigenvectors rather than just spiked ones. Let $\bx_i = \bGamma_p' \by_i$. So ${\bx_i}$ has mean zero and diagonal covariance matrix $\bLambda_p$.
Since under orthonormal transformations of the data, the empirical eigenvalues of sample covariance are invariant and the empirical eigenvectors are equivariant, the analysis will be done on ${\bx_i}$'s which naturally extends to our original data ${\by_i}$'s by a simple affine transformation. The following assumption on $\bx_i$ is imposed.

\begin{assum}[Sub-Gaussian distribution]
\label{assump2prime}
Let $\bz_{i} = \bLambda_p^{-1/2} \bx_i$ be the standardized version of the transformed data $\bx_i$. $\bz_i$'s are iid samples of sub-Gaussian isotropic random vector $\bz$, i.e., $\|\bz\|_{\phi_2} = \sup_{\bu \in \mathcal S^{p-1}} \|\langle \bz, \bu \rangle \|_{\phi_2} \le M$ for some constant $M > 0$ where the sub-Gaussian norm is defined as $\|\langle \bz, \bu \rangle\|_{\phi_2} = \sup_{p \ge 1} p^{-1/2} (\mathbb E|\langle \bz, \bu \rangle|^p)^{1/p}$. Furthermore, we assume $\exists M_1, M_2 > 0$ such that for $0 \le \theta \le M_1$,
\beq \label{indApproxCond}
\mathbb E\Big[ \exp \Big(-\theta \sum_{j=1}^p (z_{j}^2 - 1) \Big)\Big] \le \exp(M_2 \theta^2 p)\,.
\eeq
\end{assum}

The above lemma require a slightly stronger condition than the classical sub-Gaussian condition for $\bz$. It has to satisfy (\ref{indApproxCond}) for technical reasons discussed in Lemma \ref{QuadForm} in Appendix \ref{secD}. This assumption is clearly satisfied if $\bz$ has independent elements of sub-Gaussian variables \citep{Ver10} although it could also hold for weakly dependent sub-Gaussian vectors.

Under this assumption, trivially the first condition in (\ref{suffCond}) holds for the sample covariance matrix $\hat\bSigma_{Y}$ of $\by_i$, i.e., $\|\hat\bSigma_{Y} - \bSigma\|_{\max} = O_P(\sqrt{\log p/n})$. We present two theoretical properties next respectively on leading empirical eigenvalues $\{\hat \lambda_j\}_{j=1}^m$ and eigenvectors $\{\hat\bxi_j\}_{j=1}^m$ of the sample covariance matrix $\hat\bSigma$ of $\bx_i$'s. These properties are useful for us to verify the remaining conditions of the high level theoretical interface described in \eqref{suffCond}.

\begin{thm}
\label{thm_eigenvalue}
Under Assumptions \ref{assump1} and \ref{assump2prime}, for $j \le m$ we have
\[
|\hat \lambda_j/\lambda_j - 1| = O_P(n^{-1/2})\,,
\]
where $\hat\lambda_j = \lambda_j(\hat\bSigma)$ is the $j$th largest eigenvalue of $\hat\bSigma$.
\end{thm}

Consider the empirical eigenvectors $\hat \bxi_j$ of $\hat\bSigma$ for $j \le m$. Each $\hat\bxi_j$ is divided into two parts $ \hat\bxi_j = ( \hat\bxi_{jA}',  \hat\bxi_{jB}')'$, where $ \hat\bxi_{jA}$ is of length $m$ corresponding to the spike component and $\hat\bxi_{jB}$ corresponds to the noise component.
\begin{thm}
\label{thm_eigenvector}
Under Assumptions \ref{assump1} and \ref{assump2prime}, for $j \le m$ we have
 \\
(i) $\|\hat\bxi_{jA} - \be_{jA}\|  = O_P(n^{-1/2})$, where $\be_{jA}$ is unit vector of length $m$;  \\
(ii) $\|\bOmega \hat \bxi_{jB}\|_{\max} = O_p(\sqrt{\log p/(np)})$ for any $\bOmega_{p\times(p-m)}$ s.t. $\bOmega' \bOmega = \bI_{p-m}$.
\end{thm}

The theorems state that under the pervasive condition that spiked eigenvalues are of order $p$, we are able to approximately recover the true leading eigenvalues and eigenvectors. In \cite{FanWan15}, the same phenomenon is observed when $\bz_i$'s are sub-Gaussian vector with independent elements. But here we do not require element-wise independence and relax the condition to any sub-Gaussian isotropic random vectors satisfying (\ref{indApproxCond}). The proofs of the above two theorems can be found in Appendix \ref{secB}.

Given the above two theorems, let us validate the second and third conditions in (\ref{suffCond}). Define $\hat\bLambda_{SG} = \diag(\hat\lambda_1,\dots, \hat\lambda_m)$ where $SG$ is short for sub-Gaussian. The second condition holds for $\hat\bLambda_{SG}$ according to Theorems \ref{thm_eigenvalue}. Note that $\hat\bSigma_{Y}$ and $\hat\bSigma$ share the same set of empirical eigenvalues.
To check the third one, let $\hat\bGamma_{SG} = (\hat\bxi_1^{(Y)}, \dots, \hat\bxi_m^{(Y)})$ be the matrix consists of the top $m$ leading eigenvectors of $\hat\bSigma_{Y}$. If the whole eigen-space of $\bSigma$ is written as $\bGamma_p = (\bGamma, \bOmega)$, then $\hat\bxi_j^{(Y)} = \bGamma_p \hat\bxi_j =  \bGamma \hat\bxi_{jA} + \bOmega \hat\bxi_{jB}$. Therefore $\hat\bxi_j^{(Y)} - \bxi_j =  \bGamma (\hat\bxi_{jA} - \be_{jA}) +\bOmega \hat\bxi_{jB}$ and
\beq \label{eigenvector_max}
\begin{aligned}
\|\hat\bGamma_{SG} - \bGamma\|_{\max} & = \max_j \|\hat\bxi_j^{(Y)} - \bxi_j\|_{\max} \\
&\le \max_j \Big(\sqrt{m} \|\bGamma\|_{\max}\|\hat\bxi_{jA} - \be_{jA}\| + \|\bOmega\hat\bxi_{jB}\|_{\max} \Big) \,,
\end{aligned}
\eeq
which is $O_P(\sqrt{\log p/(np)})$ due to Theorem \ref{thm_eigenvector} and the fact $\|\bGamma\|_{\max} = O(1/\sqrt{p})$ shown in Theorem \ref{suff}. Hence, we have shown that the sample covariance based estimators $\hat\bSigma_{Y}, \hat\bLambda_{SG}$ and $\hat\bGamma_{SG}$  satisfy the sufficient conditions (\ref{suffCond}). Together with Theorem \ref{suff}, this explains why POET achieves all the desired rates (\ref{rate1}) and (\ref{rate2}).

We finally devote a remark to the assumption of zero mean of the observed data implied by Assumption \ref{assump2prime}. This condition  is only made to simplify the presentation of proofs. In practice, we first center the data by $\bar\by = n^{-1} \sum_i \by_i$. All the conclusions of this section hold for the centered data as well.

\section{Elliptical factor models} \label{sec4}
In the previous section, we assume $\bx_i$ to be a sub-Gaussian random vector, which is a strong distributional assumption for many applications. In this section, we  replace the sub-Gaussian assumption \ref{assump2prime} by elliptical distribution assumption \ref{assump2} and propose a novel robust estimator for the analysis of factor models.

We first briefly review the elliptical distribution family, which generalize the multivariate normal distribution and multivariate t-distribution. Compared to the sub-Gaussian setting, it is more challenging to design pilot estimators to simultaneously satisfy the three requirements in (\ref{suffCond}). To handle this challenge, we separately construct two estimators $\hat\bSigma_1$ and $\hat\bSigma_2$. $\hat\bSigma_1$ and its leading eigenvalues satisfies the first two requirements in (\ref{suffCond}) while the eigenvectors of $\hat\bSigma_2$ satisfies the last condition of (\ref{suffCond}).


\subsection{Elliptical distribution}
We define the elliptical distribution as follows. Let $\bmu \in \mathbb R^p$ and $\bSigma \in \mathbb R^{p \times p}$ with $\rank(\bSigma) = q \le p$. A $p$-dimensional random vector $\by$ has an elliptical distribution, denoted by $\by \sim EC_p(\bmu, \bSigma, \zeta)$, if it has a stochastic representation
\beq \label{Eq:ellip}
\by \overset{d} = \bmu + \zeta \bA \bU\,,
\eeq
where $\bU$ is a uniform random vector on the unit sphere in $\mathbb R^q$, $\zeta \ge 0$ is a scalar random variable independent of $\bU$, $\bA \in \mathbb R^{p \times q}$ is a deterministic matrix satisfying $\bA \bA' = \bSigma$. Here $\bSigma$ is called the
scatter matrix. Note that the representation in \eqref{Eq:ellip} is not identifiable since we can rescale $\zeta$ and $\bA$.  To make the model identifiable, we require $\mathbb E\zeta^2 = q$ so that $\Cov(\bY) = \bSigma$. In addition, we assume  $\bSigma$ is non-singular,  i.e., $q=p$. If $q < p$, as long as they are of the same order, all results in the following still hold. In this paper, we only consider continuous elliptical distributions with $\mathbb P(\zeta = 0) = 0$.

An equivalent definition of an elliptical distribution is through its characteristic function, which admits the form
$\exp(i \bt' \mu) \psi(\bt' \bSigma \bt)$, where $\psi$ is a properly defined characteristic function and $i := \sqrt{-1}$. $\zeta$ and $\psi$ are mutually determined by each other. In this setting, we denote by $\by \sim EC_p(\bmu, \bSigma, \psi)$. The marginal and conditional distributions of an elliptical distribution are also elliptical. Therefore, in factor model (\ref{eq:factor}), if $\bff_t$ and $\bu_t$ are uncorrelated and jointly elliptical, i.e., $(\bff_t', \bu_t')' \sim EC_p({\bf 0}, \diag(\bI_m, \bSigma_u), \zeta)$, then we have $\by_t \sim EC_p({\bf 0}, \bSigma, \zeta)$.

Compared to the Gaussian family, the elliptical family provides more flexibility in modeling complex data. The main advantage of the elliptical family is its ability to model heavy-tail data  and the tail dependence between variables \citep{HulLin02}, which makes it useful for modeling many modern datasets, including financial data \citep{Rac03, CizHarWer05}, genomics data \citep{LiuHawGhoYou03, PosFelSyk11}, and fMRI brain-imaging data \citep{Rut98}.

The following assumption is considered in this section.

\begin{assum}[Elliptical distribution]
\label{assump2}
The data $\by_i$'s are elliptically distributed, i.e., ${\by_i} \sim EC_p(\bmu, \bSigma, \zeta)$ or $\by_i \overset{d} = \bmu + \zeta_i \bSigma^{\frac 1 2} \bU_i$ with $\bU_i$ uniformly distributed on the unit sphere $\mathcal S^{p-1}$ and the random variable $\zeta_i \ge 0$ independent from $\bU_i$. Additionally, we assume $\mathbb E[\zeta_i^2] = p$ due to identifiability and  $\max_{j \le p} \mathbb E y_{ij}^4$ is bounded.
\end{assum}

The above assumption is implied by imposing a joint elliptical model of the factors and noises, i.e., $(\bff_t', \bu_t')' \sim EC_p({\bf 0}, \diag(\bI_m, \bSigma_u), \zeta)$. Obviously, the elliptical family is more general than Gaussian assumption and contains heavy tail distributions.  One typical example is multivariate t-distribution with degrees of freedom $\nu > 4$. The moment condition is imposed only for the sake of estimating marginal variances by methods discussed in Section \ref{sec4.1}. This assumption may be relaxed if other methods are applied.

\subsection{Robust estimation of variances} \label{sec4.1}

Let $\bSigma = \bD\bR\bD$ where $\bR$ is the correlation matrix and $\bD = \diag(\sigma_1,\dots, \sigma_p)$ is the diagonal matrix consists of standard deviations for each dimension. Our construction of $\hat\bSigma_1$ is based on separately estimating  $\bD$ and $\bR$. In this subsection, we first introduce a robust estimator $\hat\bD$ to estimate $\bD$.

Since $\by_i \sim EC_p(\bmu, \bSigma, \zeta)$ exhibits heavy tails, we need a method to robustly estimate $\bmu$ in order to center the data and estimate the covariance matrix. Substantial amount of research has been conducted on this subject in both low dimensional setting \citep{Hub64, ZouYua08, WuLiu09} and high dimensional setting \citep{BelChe11, FanFanBar14}. In addition, \cite{Koe05} has considered problem from a quantile regression perspective.  In this section, we introduce two M-estimator methods proposed by \cite{FanLiWan14} and \cite{Cat12}, who borrow the original idea from \cite{Hub64}. The methods are also useful for robust estimation of variances.

Let us denote $\bmu = (\mu_1, \dots, \mu_p)'$ and $\by_i = (y_{i1},\dots ,y_{ip})'$ for $i = 1, \dots, n$. We estimate each $\mu_j$  using the data $\{y_{1j} ,\dots , y_{nj}\}$. The M-estimator $\hat\mu = (\hat\mu_1, \dots, \hat\mu_p)'$ of \cite{FanLiWan14} is obtained by solving
\beq \label{mEstimator}
\sum_{i=1}^n h[\alpha (y_{ij} - \hat\mu_j)] = 0\,
\eeq
for each $j \le p$, where $h: \mathbb R \to \mathbb R$ is the derivative function of the Huber loss satisfying $h(x) = x$ if $|x| \le 1$, $h(x) = 1$ if $x > 1$ and $h(x) = -1$ if $x < -1$.
 The above estimator can be equivalently obtained by minimizing the Huber loss
\[
\ell_\alpha(x) = \left\{
  \begin{array}{lr}
    2 \alpha^{-1} |x| - \alpha^{-2} & : |x| > \alpha^{-1};\\
    x^2 & : |x| \le \alpha^{-1}.
  \end{array}
\right.
\]
According to \cite{FanLiWan14}, choosing $\alpha = \sqrt{\log(\epsilon^{-1})/(nv^2)}$ for $\epsilon \in (0,1)$ such that $\log(\epsilon^{-1}) \le n/8$, where $v$ is an upper bound of $\max\{\sigma_1^2, \dots, \sigma_p^2\}$, we have
\beq
\mathbb P\Big(|\hat\mu_j - \mu_j| \le 4v\sqrt{\frac{\log (\epsilon^{-1})}{n}} \Big) \ge 1-2\epsilon \,.
\eeq

\cite{Cat12} proposed another  M-estimator by solving
(\ref{mEstimator}) with a different strictly increasing $h(x)$ such that $- \log(1-x +x^2/2) \le h(x) \le \log(1+x +x^2/2)$. For a value $\epsilon \in (0,1)$ such that $n > 2 \log (1/\epsilon)$, let
\[
\alpha = \sqrt{\frac{2\log(\epsilon^{-1})}{n(v + \frac{2v\log(\epsilon^{-1})}{n-2\log(\epsilon^{-1})})}}\,,
\]
where $v$ is again an upper bound of $\max\{\sigma_1^2, \dots, \sigma_p^2\}$. \cite{Cat12} showed that the solution of (\ref{mEstimator}) satisfies
\beq
\mathbb P\Big(|\hat\mu_j - \mu_j| \le \sqrt{\frac{2v\log (\epsilon^{-1})}{n-2\log(\epsilon^{-1})}} \Big) \ge 1-2\epsilon \,.
\eeq

Therefore, by taking $\epsilon = 1/(n \vee p)^2$, $|\hat\bmu -\bmu|_{\infty} \le C\sqrt{\log p/n}$ with probability at least $1 - 2(n \vee p)^{-1}$ for both methods. We implement Catoni's estimator in the simulation by taking $h(x) = \sgn(x) \log(1 + |x| + x^2/2)$. For the choice of $v$, we simply take $v = \max\{\tilde\sigma_1^2, \dots, \tilde\sigma_p^2\}$, where $\tilde\sigma_j^2$ are the sample covariance of the $j$th dimension.

To estimate $\sigma_j^2$, we apply the above M-estimation methods on the squared data.  Note that $\sigma_j^2 = \mathbb E(Y_{ij}^2) - \mu_j^2$. We have estimated $\mu_j$ above. To estimate $\mathbb E(Y_{ij}^2)$, we employ the M-estimator (\ref{mEstimator}) on the squared data $\{y_{1j}^2, \dots, y_{nj}^2\}$, denoted by $\hat\eta_j$. This works as the fourth moment of $y_{ij}$ is assumed finite. The robust variance estimator is then defined as 
\beq
\hat\sigma_j^2 = \max\{ \hat\eta_j - \hat\mu_j^2, \delta_0\}\,,
\eeq
where $\delta_0 > 0$ is a small constant ($\delta_0 < \min\{\sigma_1^2, \dots, \sigma_p^2\}$). If $n \ge C \log d$, let $\hat\bD = \diag(\hat\sigma_1, \dots, \hat\sigma_p)$, we have
\beq \label{eq:D}
\|\hat \bD - \bD \| = O_P(\sqrt{\log p/n}).
\eeq
Additionaly, due to the structure of $\bSigma = \bB'\bB + \bSigma_u$, where $\|\bSigma_u\| \le C$ and $\|\bB\|_{\max} \le C$, it is easy to see $\|\bD\| = O(1)$ and $\|\hat\bD\| = O_P(1)$.

\subsection{Marginal Kendall's tau estimator} \label{sec4.2}

We now provide a pilot estimator to robustly estimate the correlation matrix $\bR = (r_{jk})$ when data follow elliptical distributions. The idea of Kendall's tau statistic was introduced by \cite{Ken48}  for estimating pairwise comovement correlation. Kendall's tau correlation coefficient is defined as
\beq
\hat\tau_{jk} := \frac{2}{n(n-1)} \sum_{i < i'} \sgn((Y_{ij} - Y_{i'j})(Y_{ik} - Y_{i'k})) \,,
\eeq
whose population counterpart is
\beq
\tau_{jk} := \mathbb P((Y_{1j} - Y_{2j})(Y_{1k} - Y_{2k}) > 0) -  \mathbb P((Y_{1j} - Y_{2j})(Y_{1k} - Y_{2k}) < 0)\,.
\eeq
Note that the estimator does not depend on the location $\bmu$. So without loss of generality, we assume $\bmu = \bf 0$. Then $\by \sim EC_p({\bf 0}, \bSigma, \zeta)$ with independent and identically distributed samples $\by_1, \dots, \by_n$.

Denote by $\bT = (\tau_{jk})$ and $\hat\bT = (\hat\tau_{jk})$. For elliptical family, it is known that the nonlinear relationship $r_{jk} = \sin(\frac{\pi}{2} \tau_{jk})$ holds for the Pearson correlation and Kendall's correlation \citep{FanKotNg90, HanLiu14}. Therefore, a natural estimator for $\bR$ is $\hat\bR = (\hat r_{jk})$ where
\beq
\hat r_{jk} = \sin\Bigl(\frac{\pi}{2} \hat\tau_{jk}\Bigr)\,.
\eeq

By Theorem 3.2 of \cite{HanLiu13a}, with probability larger than $1-2\epsilon - \epsilon^2$ for any $\epsilon \in (0,1)$,
\begin{align*}
\lefteqn{\|\hat\bR -\bR\|_2}  \\
 && \le  \pi^2 \|\bR\|_2 \Big( 2\sqrt{\frac{(\tr{\bR}/\|\bR\|_2+1)\log(p/\epsilon)}{3n}} + \frac{(\tr{\bR}/\|\bR\|_2+1)\log(p/\epsilon)}{n} \Big) \,.
\end{align*}
Using the fact $\|\bD\|^{-2} \|\bSigma\| \le \|\bR\| \le \|\bD^{-1}\|^2\|\bSigma\| $, we know $\|\bR\| \asymp \|\bSigma\| \asymp p$ since all the eigenvalues of $\bD$ are bounded away from infinity and zero. This is true because $\lambda_{\min}(\bD^2) \ge \lambda_{\min}(\bSigma) \ge c_0$ and $\|\bD\| = O(1)$ as derived in Section \ref{sec4.1}. This implies
\beq \label{eq:R}
\|\hat\bR -\bR\|_2 = O_P\Big(\sqrt{\frac{p^2 \log p}{n}}\Big)\,.
\eeq
Combining the rates in (\ref{eq:D}) and (\ref{eq:R}), we conclude
\begin{eqnarray*}
\lefteqn{\Big|\lambda_j(\hat\bD\hat\bR\hat\bD) - \lambda_j(\bD\bR\bD) \Big| \le \|\hat\bD\hat\bR\hat\bD - \bD\bR\bD \|} \\
 &= &  O_P(\|(\hat\bD - \bD) \bR\bD\| + \|\hat\bD (\hat\bR - \bR) \hat\bD\|) \\
 &=  & O_P\Big(\sqrt{\frac{p^2\log p}{n}}\Big) \,.
\end{eqnarray*}
Define $\hat\bSigma_1 = \hat\bD\hat\bR\hat\bD$. The estimator $\hat\bLambda_{ED} = \diag(\lambda_1(\hat\bSigma_1), \dots, \lambda_m(\hat\bSigma_1))$, which consists the first $m$ eigenvalues of $\hat\bSigma_1$, satisfies
\beq
\|(\hat\bLambda_{ED} - \bLambda)\bLambda^{-1}\| = O_P\Big(\sqrt{\frac{\log p}{n}}\Big)\,.
\eeq
Here $ED$ is short for elliptical distribution. This makes the second sufficient condition in (\ref{suffCond}) hold. Furthermore, we can easily check that the first sufficient condition holds for $\hat\bSigma_1$ using the concentration of U-statistics, i.e.
\beq
\|\hat\bSigma_1 - \bSigma \|_{\max} = O_P(\sqrt{\log p/n})\,.
\eeq

Although the marginal Kendall's tau based estimator $\hat\bSigma_1 $ has good properties for eigenvalues, it is hard to prove the third sufficient condition for eigenvectors in \eqref{suffCond} due to the complicated nonlinear $\sin(\cdot)$ transformation. Luckily, we do not require $\hat\bGamma$ and $\hat\bLambda$ in \eqref{suffCond}  to come from the same covariance estimator. In the next section, we propose another covariance estimator $\hat\bSigma_2$ whose eigenvectors satisfy the  third sufficient condition  in \eqref{suffCond}.

\subsection{Multivariate Kendall's tau estimator} \label{sec4.3}

To find an estimator $\hat\bGamma_{ED}$ that satisfies the third condition in (\ref{suffCond}),
we resort to the multivariate Kendall's tau estimator. We focus our analysis again on the transformed data $\bx_i = \bGamma_p' \by_i$.
The population multivariate Kendall's tau matrix is defined as
\beq \label{eq4.13}
\bK := \mathbb E\Big( \frac{(\bx_1 - \bx_2)(\bx_1 - \bx_2)'}{\|\bx_1 - \bx_2\|_2^2} \Big) \,.
\eeq
The sample version of the multivariate Kendall's tau estimator is a second-order U-statsitc:
\beq\label{eq::mkendall}
\hat\bK :=  \frac{2}{n(n-1)} \sum_{i < i'} k(\bx_i, \bx_{i'})\,,
\eeq
where
\beq
k(\bx_i, \bx_{i'}) = \frac{(\bx_i - \bx_{i'})(\bx_i - \bx_{i'})'}{\|\bx_i - \bx_{i'}\|_2^2}\,. \nonumber
\eeq

Several important properties of the above estimator is worth mentioning. First this estimator is location invariant, which allows us to assume $\bmu = \bf 0$ without generality, i.e., $\bx_i \overset{d} = \zeta_i \bLambda_p^{\frac12} \bU_i$. Secondly, the eigenvectors of the estimator $\hat{\bK}$ is equivariant to orthogonal transformation. 
So if we define the multivariate Kendall's tau estimator based on the observed data $\by_i$ as
$$
    \hat\bSigma_2 = \frac{2}{n(n-1)} \sum_{i < i'} k(\by_i, \by_{i'}) = \bGamma_p \hat\bK \bGamma_p',
$$
we have $\hat\bxi_j^{(Y)} = \bGamma_p \hat\bxi_j$, where $\hat\bxi_j$ and $\hat\bxi_j^{(Y)}$ are the $j^{th}$ empirical eigenvector of $\hat\bK$ and $\hat\bSigma_2$, respectively.

The most important feature of the U-statistic estimator in \eqref{eq::mkendall} is that
 its kernel $k(\bx_i, \bx_{i'})$  is distribution-free.  To see this, we have
$$
    \bX - \widetilde \bX \overset{d}= \zeta \bLambda_p^{\frac12} \bU -\tilde \zeta \bLambda_p^{\frac12} \widetilde \bU \overset{d} = \bar \zeta \bLambda_p^{\frac12} \bU,
$$
where $\widetilde \bX$ is an independent copy of $\bX$ and the characteristic function of $\bar \zeta$ is determined by that of $\zeta$. See \cite{HulLin02} for the detailed expression of the characteristic function. Thus,
\[
k(\bX, \widetilde\bX) = \frac{(\bX - \widetilde\bX)(\bX - \widetilde\bX)'}{\|\bX - \widetilde\bX\|_2^2} \overset{d} = \frac{\bLambda_p^{\frac12} \bU \bU' \bLambda_p^{\frac12}}{\bU' \bLambda_p \bU} \overset{d} = \frac{\bLambda_p^{\frac12} \bg \bg' \bLambda_p^{\frac12}}{\bg' \bLambda_p \bg} \,,
\]
which  depends only on the multivariate standard normal vector $\bg$. The last equality is due to $\bU \overset{d} = \bg/\|\bg\|$. Thus $\bK$ defined by \eqref{eq4.13} is a diagonal matrix by the symmetry of $\bg$.

Write $\bK = \diag(\theta_1, \dots, \theta_p)$, where $\theta_j$ is defined as
\[
\theta_j = \mathbb E\biggl( \frac{\lambda_j g_{ij}^2}{\sum_{k=1}^p \lambda_k g_{ik}^2} \biggr)\, ,
\]
which is a multiple of $\lambda_j$.
Obviously, $\bK$ shares the same eigenvalue ordering as that of $\Cov(\bx_i) = \bLambda_p$, and thus the same eigenspaces as those of $\Cov(\bx_i)$. So estimating the leading eigenvectors of $\Cov(\bx_i)$ is equivalent to estimating those of $\bK$. In sum, $\hat\bSigma_2$ particularly fits the goal of  estimating the eigenvectors of $\bSigma$.

The above multivariate Kendall's tau statistic is first introduced in \cite{ChoMar98} and has been used for low dimensional covariance estimation \citep{VisKoiOja00} and principal component estimation \citep{Mar99, CroOllOja02}. Many testing literature based on rank statistics is also related to the estimator, for example \cite{Tyl82, HalPai06}. The literature listed here is only illustrative rather than complete.

We now consider the theoretical properties of the eigenvectors $\hat\xi_j$ of $\hat\bK$. As before, $\hat\bxi_j$ is divided into the spiked part $\hat\bxi_{jA}$ and noise part $\hat\bxi_{jB}$.

\begin{thm}
\label{thm_eigenvectorHT}
Under Assumptions \ref{assump1} and \ref{assump2}, for $j \le m$ we have \\
(i) $\| \hat\bxi_{jA} - \be_{jA}  \| = O_P(n^{-1/2})$, where $\be_{jA}$ is a unit vector of length $m$;  \\
(ii) $\|\bOmega \hat \bxi_{jB}\|_{\max} = O_P(\sqrt{\log p/(np)})$ for any $\bOmega_{p\times(p-m)}$ s.t. $\bOmega' \bOmega = \bI_{p-m}$.
\end{thm}

The proof for Theorem \ref{thm_eigenvectorHT} is relegated to Appendix \ref{secC}. Define $\hat\bSigma_2$ as the multivariate Kendall's tau estimator of the observed data $\by_i$'s  and $\hat\bGamma_{ED} = (\hat\bxi_1^{(Y)}, \dots, \hat\bxi_m^{(Y)})$ as the leading eigenvectors of $\hat\bSigma_2$. Theorem \ref{thm_eigenvectorHT} implies $$\|\hat\bGamma_{ED} - \bGamma\|_{\max} = O_P(\sqrt{\log p/(np)})$$ following the same argument (\ref{eigenvector_max}) in Section \ref{sec3}. So the third sufficient condition in (\ref{suffCond}) holds for $\hat\bGamma_{ED}$. Together with the estimators $\hat\bSigma_1$ and $\hat\bLambda_{ED}$ defined in Section \ref{sec4.2}, we are ready to apply the general POET procedure for the heavy tail factor model and achieve all the desired estimation convergence rates for both covariance and precision matrices.

\section{Simulations} \label{sec5}

Simulations are carried out in this section to demonstrate the effectiveness of the proposed method for elliptical factor models. The robust estimators $\hat\bSigma_1$, $\hat\bLambda_{ED}$, $\hat\bGamma_{ED}$ proposed in Section \ref{sec4} will be compared with the original POET estimator based on the sample covariance, or $\hat\bSigma_Y$, $\hat\bLambda_{SG}$, $\hat\bGamma_{SG}$ discussed in Section \ref{sec3}. We put the two sets of estimators into the general POET framework described in Section \ref{sec2} for estimating both conditional sparsity covariance and conditional graphical models.

\subsection{Conditional sparse covariance estimation}\label{sec::sim1}

In this section, we consider the factor model (\ref{eq:factor}) with $(\bff_t, \bu_t)$ jointly follow a multivariate t-distribution with degrees of freedom $\nu$. Larger $\nu$ corresponds to lighter tail and $\nu = \infty$ corresponds to a multivariate normal distribution. We simulated $n$ independent samples of $(\bff_t, \bu_t)$ from multivariate t-distribution with covariance matrix $\diag(\bI_m, \bI_p)$ and each row of $\bB$ from $\mathcal N({\bf 0}, \bI_m)$. The observed data is formed as $\by_t = \bB \bff_t + \bu_t$ and the true covariance is $\bSigma = \bB \bB' + \bI_p$. We vary $p$  from $100$ to $1000$ with sample size $n = p/2$, and fixed number of factors $m = 3$ in this simulation.

For each triple $(p, n, m)$, both the original POET estimator ($\hat\bSigma_Y$, $\hat\bLambda_{SG}$, $\hat\bGamma_{SG}$) and the proposed robust POET estimator ($\hat\bSigma_1$, $\hat\bLambda_{ED}$, $\hat\bGamma_{ED}$) were employed to estimate $\bSigma_u$ and $\bSigma$. $100$ simulations were conducted for each case. The log-ratio (base 2) of the average estimation errors using the two methods were reported in Figure \ref{fig:POET}, measured under different norms ($\|\hat\bSigma_u^{\top} - \bSigma_u\|_{2}$ and $\|(\hat\bSigma_u^{\top})^{-1} - \bSigma_u^{-1}\|_2$ for $\bSigma_u$; $\|\hat\bSigma^{\top} - \bSigma\|_{\max}, \|\hat\bSigma^{\top} - \bSigma\|_{\bSigma}$ and $\|(\hat\bSigma^{\top})^{-1} - \bSigma^{-1}\|_2$ for $\bSigma$;  $\|\hat\bSigma - \bSigma\|_{\max}, \|(\hat\bLambda - \bLambda)\bLambda^{-1}\|_2$ and $\|\hat\bGamma - \bGamma\|_{\max}$ for initial pilot estimators). In addition, three different degrees of freedom $\nu = 4.2, \nu = 7, \nu = \infty$ were chosen, representing respectively  heavy tail, moderate heavy tail, and normal situations.

From Figure \ref{fig:POET}, when factors and noises are heavy-tailed from $t_{4.2}$ (black), the original POET estimators are poorly behaved while the robust method works well as we expected. $t_7$ (blue) typically fits financial or biological data better than normal in practice. In this case, we also observe a significant advantage of the robust POET estimators. The error is roughly reduced by a magnitude of two if the rank based estimation is applied. However, when the distribution is indeed normal or $t_{\infty}$ (orange), the original POET estimators based on sub-Gaussian data performs better, though the robust POET also achieves comparable performance.

\begin{figure}
	\centering
      	\includegraphics[width=0.8\textwidth]{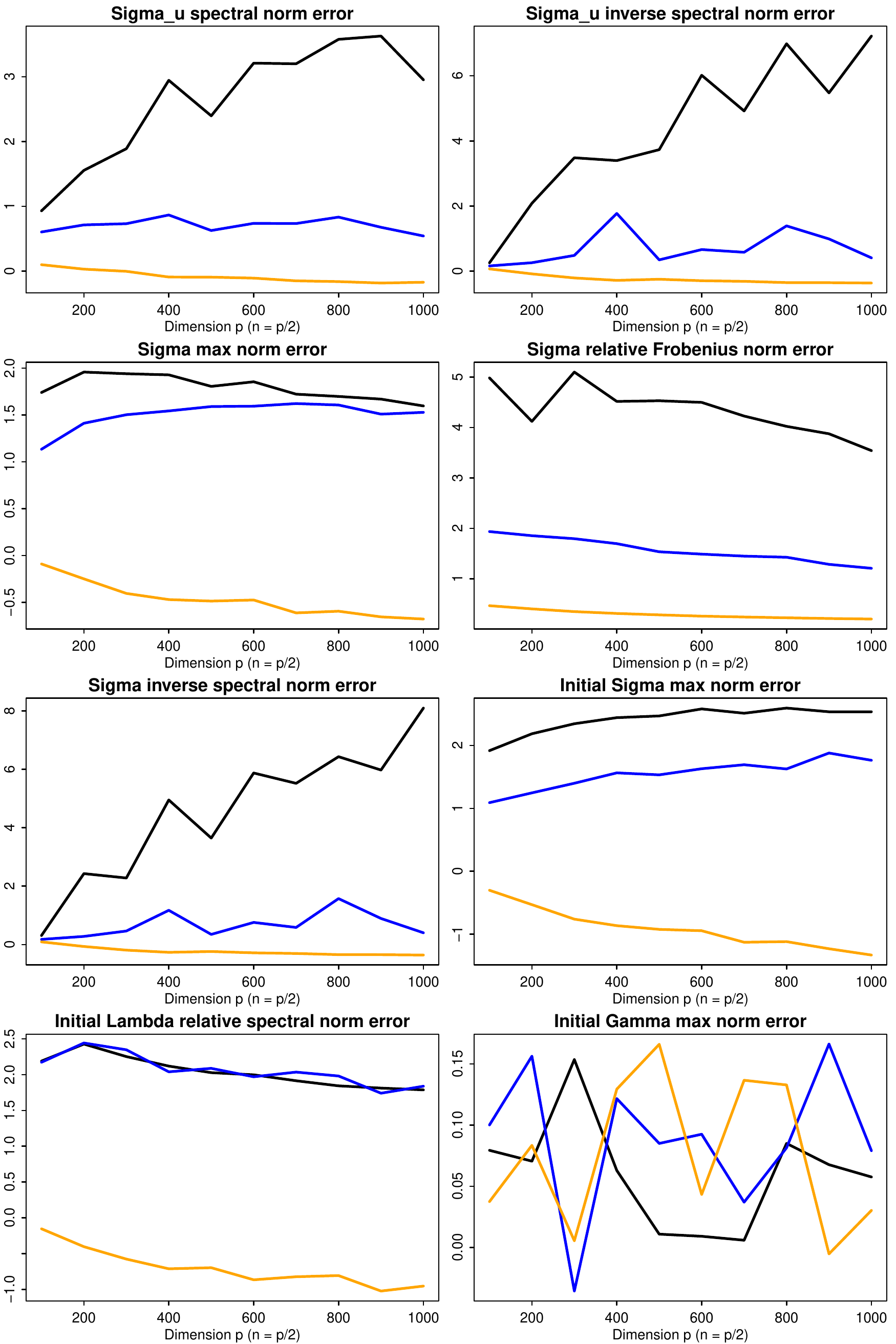}
        \caption{Conditional sparse covariance matrix estimation. The $8$ plots corresponds to logarithms (base 2) of the ratios  of average errors of the original and the robust POET estimators, measured in different norms. Data were generated from multivariate t-distribution with degree of freedom $\nu = 4.2$ (black), $\nu = 7$ (blue), $\nu = \infty$ (orange) with $p$ from $100$ to $1000$, $n = p/2$ and $m = 3$. $100$ simulations were conducted for each $p$. }
        \label{fig:POET}
\end{figure}

\subsection{Conditional graphical model estimation}

In this section, we consider the conditional graphical model described in Section \ref{sec2.3}.  In particular, we compare the accuracy of different methods for estimating the precision matrices $\bOmega_u$ and $\bOmega$. Here we assume a block diagonal precision error matrix  $\bOmega_u = \diag(\bM, \dots, \bM)$ where $\bM$ is $2$ by $2$ correlation matrix with off-diagonal element equals $0.5$. Then we simulate $(\bff_t, \bu_t)$ again from multivariate t-distribution with covariance $\diag(\bI_m, \bOmega_u^{-1})$. We set the dimension $p$ to range from $50$ to $500$, sample size $n = 0.6p$ and a  fixed number of factors $m = 3$.

For each configuration of $(p, n, m)$, after applying POET with the original and robust  pilot estimators,
we estimate $\hat\bOmega_u$ and $\hat\bOmega$ as proposed in Section \ref{sec2.3} using the CLIME procedure. To efficiently solve large-scale CLIME optimization (\ref{CLIME}), we used the R package ``fastclime'' developed by \cite{PanLiuVan14}. $100$ simulations were conducted for each case. The log-ratio (base 2) of the average errors of the two methods were reported in Figure \ref{fig:Graph}, measured under spectral norms $\|\hat\bOmega_u - \bOmega_u\|_{2}$ and $\|\hat\bOmega - \bOmega\|_{2}$. Three different degrees of freedom $\nu = 4.2$ (black), $\nu = 7$ (blue), $\nu = \infty$ (orange) were used as in Section \ref{sec::sim1}. Clearly, the robust estimators outperform non-robust ones for $t_{4.2}$ and $t_7$, and maintains competitive for the normal case.

\begin{figure}
	\centering
      	\includegraphics[width=0.8\textwidth]{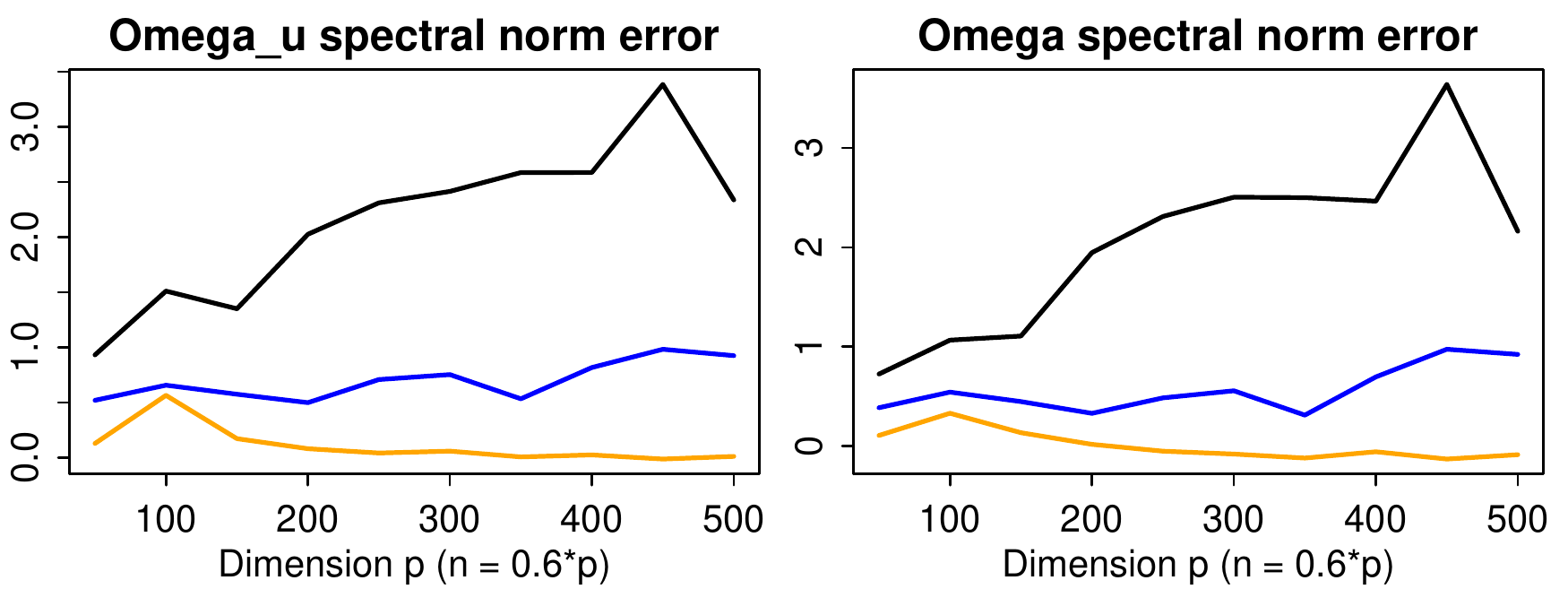}
        \caption{Conditional graphical model estimation. The plots corresponds to log ratio (base 2) of average errors of the original and the robust POET estimators for $\hat\bOmega_u$ and $\hat\bOmega$, measured in spectral norms. Data were generated from multivariate t-distribution with degree of freedom $\nu = 4.2$ (black), $\nu = 7$ (blue), $\nu = \infty$ (orange) with $p$ from $50$ to $500$, $n = 0.6p$ and $m = 3$. $100$ simulations were conducted for each $p$. }
        \label{fig:Graph}
\end{figure}

\section{Discussions} \label{sec6}

We provide a fundamental understanding of high dimensional factor models under the pervasive condition. In particular, we extend the POET estimator in \cite{FanLiaMin13}  to  be a generic procedure which could take any pilot covariance matrix estimators as initial inputs,  as long as they satisfy a set of sufficient conditions specified in (\ref{suffCond}). Transparent theoretical results are then developed. The main challenge is to check the high level conditions hold for certain estimators. When the observed data $\by_i$ is sub-Gaussian random vector, we are able to simply use sample covariance matrix to construct initial estimators. However, if we encounter heavy-tailed elliptical distributions, robust estimators for eigen-structure should be considered. The paper provides an example of separately estimating leading eigenvalues and eigenvectors under elliptical factor models. But the results could possibly be generated to other richer family of distributions.

Based on recent work of \cite{FanWan15}, it is possible to relax the spiked eigenvalue condition from order $p$ to weaker signal level. But bounded eigenvalues are obviously not enough for consistent estimation as has been pointed out by \cite{JohLu09}, and as a result more structural assumptions are needed. \cite{AgaNegWai12} considered a similar type of low-rank plus sparse decomposition, but their work is based on optimization technique and does not leverage pervasiveness. In addition, they
only analyze the obtained estimator using Frobenius norm errors. Consequently, their lower bound results are not applicable to our setting. The lower bound for the rates in (\ref{rate1}) and (\ref{rate2}) will be pursued in a separate work. By all means, the similarity and difference between optimization thinking and pervasiveness thinking should be studied in further details.

\appendix
\section{Proofs in Section 2} \label{secA}
\begin{proof}[\bf Proof of Theorem \ref{suff}]
We establish (\ref{rate1}) here and defer the details of the proof of (\ref{rate2}) to Appendix \ref{secA}..
To obtain the rates of convergence in (\ref{rate1}), it suffices to prove $\|\hat\bSigma_u - \bSigma_u\|_{\max} = O_P(w_n)$. Once the max error of sparse matrix $\bSigma_u$ is controlled, it is not hard to show the adaptive procedure discussed in (\ref{thresholding}) gives $\hat\bSigma_u^{\top}$ such that the spectral error $\|\hat\bSigma_u^{\top} - \bSigma_u\|_{2}  = O_P( m_p w_n^{1-q} )$ \citep{FanLiaMin11, CaiLiu11, RotLevZhu09}. Furthermore, $\|(\hat\bSigma_u^{\top})^{-1} - {\bSigma_u}^{-1}\|_{2} \le \|(\hat\bSigma_u^{\top})^{-1} \|_2 \|\hat\bSigma_u^{\top} - \bSigma_u\|_2 \|\bSigma_u^{-1}\|_2$. So $\|(\hat\bSigma_u^{\top})^{-1} - {\bSigma_u}^{-1}\|_{2}$ is also $O_P( m_p w_n^{1-q} )$ due to the lower boundedness of $\|\bSigma_u\|_2$.

According to first condition in (\ref{suffCond}), $\|\hat\bSigma - \bSigma\|_{\max} = O_P(\sqrt{\log p/n})$. Therefore to show $\|\hat\bSigma_u - \bSigma_u\|_{\max} = O_P(w_n)$, we only need to prove the low rank part of $\bSigma$  concentrates at a desired rate under max norm. So our goal is to prove
\beq \label{goal}
\|\hat\bGamma\hat\bLambda\hat\bGamma - \bB\bB'\|_{\max} = O_P(\sqrt{\log p/n} + 1/\sqrt{p})\,.
\eeq

Let $\bB\bB' = \tilde \bGamma \tilde \bLambda \tilde \bGamma'$ where $\tilde\bLambda = \diag(\|\bb_1\|^2, \dots, \|\bb_m\|^2)$ and the $j$th column of $\tilde \bGamma$ is $\bb_j/\|\bb_j\|$. To obtain (\ref{goal}), we bound $\Delta_1 := \|\hat\bGamma\hat\bLambda\hat\bGamma' - \bGamma\bLambda\bGamma'\|_{\max}$ and $\Delta_2 := \|\tilde\bGamma\tilde\bLambda\tilde\bGamma' - \bGamma\bLambda\bGamma'\|_{\max}$ separately. Four useful rates of convergence are listed in the following:
\begin{align*}
& \|\tilde\bLambda - \bLambda\|_{\max} \le \|\bSigma_u\| = O(1)\,, \\
& \|\tilde\bGamma - \bGamma\|_{\max} \le C\|\bSigma_u\|/p = O(1/p)\,, \\
& \|(\hat\bLambda - \bLambda)\bLambda^{-1}\|_{\max} = O_P(\sqrt{\log p/n})\,, \\
& \|\hat\bGamma - \bGamma\|_{\max} = O_P(\sqrt{\log p/(np)})\,.
\end{align*}
The first one is due to Weyl's inequality since $\|\tilde\bLambda - \bLambda\|_{\max} = \|\tilde\bLambda - \bLambda\|_{2}$ while the second follows from trivial bound $\|\tilde\bGamma - \bGamma\|_{\max} \le \|\tilde\bGamma - \bGamma\|_{F}$, which is further bounded by $C\|\bSigma_u\|/p$ according to the $\sin \theta$ theorem of \cite{DavKah70}. The third and fourth rates are by assumption. Next we show $\|\bGamma\|_{\max} = O(1/\sqrt{p})$ and derive the rates for $\Delta_1$ and $\Delta_2$.

Note that
\begin{align*}
\|\bGamma\bLambda^{\frac12} - \tilde\bGamma\tilde\bLambda^{\frac12}\|_{\max} & \le \|\bB\tilde\bLambda^{-\frac12}(\bLambda^{\frac12} - \tilde\bLambda^{\frac12})\|_{\max} + \|(\bGamma - \tilde\bGamma)\bLambda^{\frac12}\|_{\max} \\
& \le C\frac{ \|\bB\|_{\max} + \|\bSigma_u\|}{\sqrt{p}} = o(1) \,.
\end{align*}
Since $\|\bB\|_{\max} = \|\tilde\bGamma\tilde\bLambda^{\frac12}\|_{\max} = O(1)$, we have $\|\bGamma\bLambda^{1/2}\|_{\max} = O(1)$ and $\|\bGamma\|_{\max} = O(1/\sqrt{p})$. Using this fact, the following argument implies $\Delta_1 = O_P(\sqrt{\log p/n})$ and $\Delta_2 = O(\sqrt{1/p})$. More specifically,
\begin{align*}
\Delta_1 & \le \|\hat\bGamma(\hat\bLambda-\bLambda)\hat\bGamma'\|_{\max} + \|(\hat\bGamma - \bGamma) \bLambda (\hat\bGamma - \bGamma)'\|_{\max} + 2\|\bGamma\bLambda(\hat\bGamma - \bGamma)'\|_{\max} \\
& = O_p(p^{-1} \|\hat\bLambda-\bLambda\|_{\max} + \sqrt{p} \|\hat\bGamma - \bGamma\|_{\max}) = O_P(\sqrt{\log p/n})\,, \\
\Delta_2 & \le \|\tilde\bGamma(\tilde\bLambda-\bLambda)\tilde\bGamma'\|_{\max} + \|(\tilde\bGamma - \bGamma) \bLambda (\tilde\bGamma - \bGamma)'\|_{\max} + 2\|\bGamma\bLambda(\tilde\bGamma - \bGamma)'\|_{\max} \\
& = O(p^{-1} \|\tilde\bLambda-\bLambda\|_{\max} + \sqrt{p} \|\tilde\bGamma - \bGamma\|_{\max}) = O(\sqrt{1/p})\,.
\end{align*}
Combining the rates of $\Delta_1$ and $\Delta_2$, we prove  (\ref{goal}). Thus (\ref{rate1}) follows.

Now let us prove (\ref{rate2}). The first  result follows from $\|\hat\bSigma_u^{\top} - \bSigma_u\|_{\max} \le \|\hat\bSigma_u^{\top} - \hat\bSigma_u\|_{\max} + \|\hat\bSigma_u - \bSigma_u\|_{\max} = O_P(\tau + w_n) = O_P(w_n)$ when $\tau$ is chosen as the same order $w_n$ and
\[
\|\hat\bSigma^{\top} - \bSigma\|_{\max} \le \|\hat\bGamma\hat\bLambda\hat\bGamma - \bB\bB'\|_{\max} + \|\hat\bSigma_u^{\top} - \bSigma_u\|_{\max} = O_P(w_n)\,.
\]

We now prove the remaining two results. Suppose the SVD decomposition of $\bSigma = \bGamma_p \bLambda_p \bGamma_p'$ where $\bGamma_p = (\bGamma, \bOmega)$ and $\bLambda_p = \diag(\bLambda, \bTheta)$.
Then obviously
\beq \label{Norm_BreakDown1}
\begin{aligned}
\|\hat\bSigma^{\top}  - \bSigma\|_{\Sigma} \le & \left. p^{-1/2} \Big( \left\| \bSigma^{-\frac12} (\hat\bGamma \hat\bLambda \hat\bGamma' - \bB \bB') \bSigma^{-\frac12} \right\|_F \right.  \\
& \left. + \|\bSigma^{-\frac12} (\hat\bSigma_u^{\top} - \bSigma_u) \bSigma^{-\frac12}\|_F \Big)  =:  \Delta_L + \Delta_S \right. ,
\end{aligned}
\eeq
and
\beq \label{Norm_BreakDown3}
\Delta_S \le p^{-1/2} \|\bSigma^{-1}\| \|\hat\bSigma_u^{\top} - \bSigma_u\|_F \le C \|\hat\bSigma_u^{\top} - \bSigma_u\|  = O_P(m_p w_n^{1-q})\,.
\eeq
It is easy to show
\beq \label{Norm_BreakDown2}
\begin{aligned}
\Delta_L &=  p^{-1/2} \left\| \left( \begin{array}{c} \bLambda^{-\frac12} \bGamma' \\ \bTheta^{-\frac12} \bOmega' \end{array} \right) (\hat\bGamma \hat\bLambda \hat\bGamma' - \bB \bB') \left( \begin{array}{cc}
\bGamma \bLambda^{-\frac12} & \bOmega \bTheta^{-\frac12} \end{array} \right) \right\|_F  \\
& \le \Delta_{L1} + \Delta_{L2} + 2\Delta_{L3} \,,
\end{aligned}
\eeq
where $\Delta_{L1} = \|\bLambda^{-\frac12} \bGamma' (\hat\bGamma\hat\bLambda \hat\bGamma' - \bB \bB') \bGamma \bLambda^{-\frac12}\|_F/\sqrt{p} $, $\Delta_{L2} = \|\bTheta^{-\frac12} \bOmega' (\hat\bGamma \hat\bLambda \hat\bGamma' - \bB \bB') \bOmega \bTheta^{-\frac12} \|_F /\sqrt{p}$ and $\Delta_{L3} =  \|\bLambda^{-\frac12} \bGamma' (\hat\bGamma\hat\bLambda \hat\bGamma' - \bB \bB') \bOmega \bTheta^{-\frac12}\|_F/\sqrt{p} $.
In order to characterize the rate of  convergence under relative Frobenius norm, we analyze the terms $\Delta_{L1}, \Delta_{L2}$ and $\Delta_{L3}$ separately.

According to Theorem 4.1 of  \cite{FanWan15}, $\Delta_{L1} \le \|\bLambda^{-\frac12} \bGamma' (\hat\bGamma\hat\bLambda \hat\bGamma' - \bB \bB') \bGamma \bLambda^{-\frac12}\|_2 = O_P(n^{-1/2})$. $\Delta_{L2}$ is bounded by
\[
p^{-1/2} \Big( \|\bTheta^{-\frac12} \bOmega' \hat\bGamma \hat\bLambda \hat\bGamma' \bOmega \bTheta^{-\frac12} \|_F + \|\bTheta^{-\frac12} \bOmega' \tilde\bGamma \tilde\bLambda \tilde\bGamma' \bOmega \bTheta^{-\frac12} \|_F \Big) =: \Delta_{L2}^{(1)} + \Delta_{L2}^{(2)}\,,
\]
where
\[
\Delta_{L2}^{(1)} \le p^{-1/2}\|\bTheta^{-1}\| \|\bOmega'(\hat\bGamma-\bGamma)\|_F^2 \|\hat\bLambda\| = O_P(\sqrt{p}\log p/n)\,,
\]
because $\|\hat\bGamma-\bGamma\|_{\max} = O_P(\sqrt{\log p/(np)})$ by assumption and $\|\hat\bLambda\| = O_P(p)$. Similarly, $\Delta_{L2}^{(2)} = O_P(1/p^{3/2})$
as $\|\bOmega'\tilde\bGamma\|_F \le \sqrt{m} \|\bOmega'\tilde\bGamma\|_2 = \sqrt{m} \|\bGamma\bGamma' - \tilde\bGamma\tilde\bGamma'\|_2 = O(\|\bSigma_u\|/p) = O_P(1/p)$ by the $\sin \theta$ Theorem \citep{DavKah70}. Finally, $\Delta_{L2} = O_P(\sqrt{p}\log p/n + 1/p^{3/2})$. Following similar arguments, $\Delta_{L3}$ is dominated by $\Delta_{L1}$ and $\Delta_{L2}$. Combining the terms $\Delta_{L1}$, $\Delta_{L2}$, $\Delta_{L3}$ and $\Delta_S$ together, we complete the proof for the relative Frobenius norm.

We now turn to analyze the spectral norm error of the inverse covariance matrix. By the Sherman-Morrison-Woodbury formula, we have
\beq
\|(\hat\bSigma^{\top})^{-1} - \bSigma^{-1}\|_2 \le \|(\hat\bSigma_u^{\top})^{-1} - \bSigma_u^{-1}\| + \Delta_{inv}\,,
\eeq
where
\beq \label{inv}
\Delta_{inv} = \|(\hat\bSigma_u^{\top})^{-1}\hat\bGamma \hat\bLambda^{\frac12} (\bI_m + \hat\bJ)^{-1} \hat\bLambda^{\frac12}\hat\bGamma' (\hat\bSigma_u^{\top})^{-1}  - \bSigma_u^{-1} \tilde\bGamma \tilde\bLambda^{\frac12} (\bI_m + \tilde\bJ)^{-1} \tilde\bLambda^{\frac12}\tilde\bGamma' \bSigma_u^{-1}\|\,,
\eeq
where $\hat\bJ = \hat\bLambda^{\frac12} \hat\bGamma' (\hat\bSigma_u^{\top})^{-1} \hat\bGamma \hat\bLambda^{\frac12}$ and $\bJ = \tilde\bLambda^{\frac12} \tilde\bGamma' \bSigma_u^{-1} \tilde\bGamma \tilde\bLambda^{\frac12} $.
The right hand side can be bounded by the terms representing the differences of the ``hat'' part $(\hat\bSigma_u^{\top})^{-1}$, $\hat\bGamma \hat\bLambda^{\frac12}$, $(\bI_m + \hat\bJ)^{-1}$ and the ``tilde'' part $(\bSigma_u^{\top})^{-1}$, $\tilde\bGamma \tilde\bLambda^{\frac12}$, $(\bI_m + \tilde\bJ)^{-1}$:
\begin{align*}
&\|((\hat\bSigma_u^{\top})^{-1}  - \bSigma_u^{-1}) \tilde\bGamma \tilde\bLambda^{\frac12} (\bI_m + \tilde\bJ)^{-1} \tilde\bLambda^{\frac12}\tilde\bGamma' \bSigma_u^{-1}\| \,, \\
&\|\bSigma_u^{-1} (\hat\bGamma \hat\bLambda^{\frac12} - \tilde\bGamma \tilde\bLambda^{\frac12}) (\bI_m + \tilde\bJ)^{-1} \tilde\bLambda^{\frac12}\tilde\bGamma' \bSigma_u^{-1}\| \,, \\
&\|\bSigma_u^{-1} \tilde\bGamma \tilde\bLambda^{\frac12} ((\bI_m + \hat\bJ)^{-1} - (\bI_m + \tilde\bJ)^{-1}) \tilde\bLambda^{\frac12}\tilde\bGamma' \bSigma_u^{-1}\| \,.
\end{align*}
The first term is $O_P(m_p w_n^{1-q})$; the second term is $O_P(p^{-\frac12} \|\hat\bGamma \hat\bLambda^{\frac12} - \tilde\bGamma \tilde\bLambda^{\frac12}\|) = O_P(w_n)$; and the third term is $O_P(p \|(\bI_m + \hat\bJ)^{-1} - (\bI_m + \tilde\bJ)^{-1}\| ) = O_P(p^{-1} \|\hat\bJ - \tilde\bJ\|) = O_P(m_p w_n^{1-q})$. Thus, $\Delta_{inv} = O_P(m_p w_n^{1-q})$, which implies $\|(\hat\bSigma^{\top})^{-1} - \bSigma^{-1}\|_2 = O_P(m_p w_n^{1-q})$. Therefore, we finish the proof of the remaining parts of (\ref{rate2}) in Theorem \ref{suff}.
\end{proof}

\begin{proof}[\bf Proof of Theorem \ref{Graph}]
From (\ref{goal}), we have $\|\hat\bSigma_u - \bSigma_u\|_{\max} = O_P(w_n)$. Next we prove that the CLIME estimator will give $\hat\bOmega_u$ such that $\|\hat\bOmega_u - \bOmega_u\|_{\max} = O_P(w_n)$. Choose $\tau \ge \|\bOmega_u\|_{\infty} \|\hat\bSigma_u - \bSigma_u\|_{\max} \ge \|\bOmega_u\hat\bSigma_u - \bI\|_{\max}$ so that the true $\bOmega$ is within the region of the constraint of (\ref{CLIME}). So
\begin{equation*}
\begin{aligned}
\|\hat\bOmega_u^1 - \bOmega_u\|_{\max} & \le \|\bOmega_u (\bI - \hat\bSigma_u \hat\bOmega_u^1)\|_{\max} + \|(\bI - \hat\bSigma_u \bOmega_u)' \hat\bOmega_u^1\|_{\max}\\
& \le \|\bOmega_u\|_{\infty} \|\bI - \hat\bSigma_u \hat\bOmega_u^1\|_{\max} + \|\hat\bOmega_u^1\|_{\infty} \|\bI - \hat\bSigma_u \bOmega_u\|_{\max}\,,
\end{aligned}
\end{equation*}
where the first term is bounded by $\tau \|\bOmega_u\|_{\infty}$ since $\hat\bOmega_u^1$ is a feasible solution of (\ref{CLIME}), and the second term is   bounded by $\tau \|\bOmega_u\|_{\infty}$ due to the optimality of $\hat\bOmega_u^1$ over $\bOmega_u$. Therefore with $\tau \asymp w_n$, we have $\|\hat\bOmega_u^1 - \bOmega_u\|_{\max} = O_P(w_n)$. It is easy to see the symmetrization step does not change the rate of $\|\hat\bOmega_u - \bOmega_u\|_{\max}$. By similar arguments as in \cite{CaiLiuLuo11}, we obtain $\|\hat\bOmega_u - \bOmega_u\|_2 = O_P(M_p w_n^{1-q})$.

From the proof of the spectral norm error of the inverse covariance matrix in Theorem \ref{suff}, we  have $\|\hat\bOmega - \bOmega\|_2 = O_P(M_p w_n^{1-q})$. It remains to prove $\|\hat\bOmega - \bOmega\|_{\max}$, which is by definition bounded by
\[
O_P(\|\hat\bOmega_u - \bOmega_u\|_{\max} + p^{-1} \|\hat\bOmega_u \hat\bGamma \hat\bLambda^{\frac12} - \bOmega_u \bGamma \bLambda^{\frac12} \|_{\max} + p^{-2} \|\hat\bJ -\bJ\|_{\max})\,,
\]
where $\hat\bJ$ and $\bJ$ are defined in (\ref{inv}). After some calculations, its dominating term is $O_P(\|\hat\bOmega_u - \bOmega_u\|_{\max} ) = O_P(w_n)$.
\end{proof}

\section{Proofs in Section 3} \label{secB}

\begin{proof} [\bf Proof of Theorem \ref{thm_eigenvalue}]
Define $\bX = (\widetilde\bx_1, \dots, \widetilde \bx_p) = (\bZ_A \bLambda_A^{1/2}, \bZ_B \bLambda_B^{1/2})$ where $\bZ_A = (\widetilde\bz_1, \dots ,\widetilde\bz_m)$, $\bZ_B = (\widetilde\bz_{m+1}, \dots, \widetilde\bz_p)$ with $\widetilde\bz_j = \widetilde \bx_j /\sqrt{\lambda_j}$ and $\bLambda_A = \diag(\lambda_1, \dots, \lambda_m)$, $\bLambda_B = \diag(\lambda_{m+1}, \dots, \lambda_p)$. In order to prove the theorem, we first define two auxillary quantities as follows and analyze them separately. Let
\[
\bA = n^{-1} \sum_{j=1}^m \lambda_j \widetilde\bz_j \widetilde\bz_j', \quad \mbox{and} \quad
\bB = n^{-1} \sum_{j=m+1}^p \lambda_j \widetilde\bz_j \widetilde\bz_j'.
\]
In addition we define
\[
   \widetilde \bSigma_{n \times n}  = \frac{1}{n} \bX \bX' =  \frac{1}{n} \sum_{j=1}^p \lambda_j \widetilde\bz_j \widetilde\bz_j' = \bA + \bB,
\]
which share the same nonzero eigenvalues with the sample covariance matrix $\hat\bSigma = n^{-1} \bX' \bX$.

We first prove $\bA$ satisfy $| \lambda_j(\bA)/\lambda_j - 1| = O_P(n^{-1/2})$.
Note that $\bA = n^{-1} \bZ_A \bLambda_A\bZ_A' $ has the same eigenvalues as matrix $\widetilde \bA = n^{-1} \bar \bZ' \bar \bZ$, where $\bar \bZ = \bZ_A \bLambda_A^{1/2}$ is an $n \times m$ matrix with independent and identically distributed rows.  Each row is a sub-Gaussian random vector  with mean ${\bf 0}$ and variance $\bLambda_A$. Therefore, we are in the low dimensional situation with fixed dimension $m$, although eigenvalues diverge. By central limit theorem, the $j^{th}$ diagonal element of $\widetilde \bA$ is of order $\lambda_j(1+O_P(n^{-1/2}))$ and the $(j,k)^{th}$ off-diagonal elements are of order $\sqrt{\lambda_j \lambda_k} O_P(n^{-1/2})$. Therefore,
$\lambda_j(\bA) = \lambda_j(\widetilde \bA) = \lambda_j (\bLambda_A^{1/2} (\bI_m + O_P(n^{-1/2})) \bLambda_A^{1/2}  ) = \lambda_j (\bLambda_A (\bI_m + O_P(n^{-1/2})))$ and furthermore
\[
\lambda_j \lambda_{\min}(\bI_m + O_P(n^{-1/2})) \le \lambda_j(\bA)  \le \lambda_j\lambda_{\max}(\bI_m + O_P(n^{-1/2})) \,.
\]
Since  dimension $m$ is fixed, we have $| \lambda_j(\bA)/\lambda_j - 1| = O_P(n^{-1/2})$.

Secondly, we have $\lambda_k(\bB)/\lambda_j = o_P(n^{-1/2})$.
By the definition of $\bB$, $\bB = n^{-1} \bZ_B \bLambda_B \bZ_B'$ where $\bZ_B$ is a $n \times (p-m)$ random matrix with independent rows of zero mean and identity covariance and $\bLambda_B = \diag(\lambda_{m+1}, \cdots, \lambda_p)$. Since each row of $\bZ_B$ is independent sub-Gaussian isotropic vector of dimension $p-m$, by Lemma \ref{Key}, choose $t = \sqrt{n}$, for any $k \le n$,
\[
\frac{1}{p-m} \lambda_k(\bZ_B \bLambda_B \bZ_B') = \frac{1}{p-m}\sum_{j=m+1}^p \lambda_j + O_P\Big(\sqrt{\frac{n}{p}}\Big) = \bar c + O_P\Big(\sqrt{\frac{n}{p}}\Big) + o(1)\,.
\]
Therefore, for $k = 1, \dots, n$,
\begin{equation*}
\frac{\lambda_k(\bB)}{\lambda_j} = \frac{n \lambda_k(\bB)}{p-m} \frac{p-m}{n \lambda_j}= O_P\Big(\frac{p-m}{n \lambda_j} \Big) = o_P(n^{-\frac 12})\,.
\end{equation*}

By Wely's Theorem, $\lambda_j(\bA) + \lambda_n(\bB) \le \hat\lambda_j \le \lambda_j(\bA) + \lambda_1(\bB).$
Therefore, combining results for $\bA$ and $\bB$, we conclude $|\hat \lambda_j/\lambda_j - 1| = O_P(n^{-1/2})$.
\end{proof}

\begin{proof} [\bf Proof of Theorem \ref{thm_eigenvector}]
(i) We start by proving the rate for $\hat\bxi_{jA}$ in the simple case of $m = 1$. In this case, by Lemma \ref{lemB.1}, we indeed have $|\hat\xi_{1A} - 1| = O_P(n^{-1/2})$ since $p > n$. In the following, we consider $m > 1$. Denote ${\bX} = (\bZ_A \bLambda_A^{1/2}, \bZ_B \bLambda_B^{1/2})$ as before.
Recall $\hat\bxi_1,\hat \bxi_2,\dots, \hat\bxi_p$ are eigenvectors of $\hat\bSigma = n^{-1} {\bX}' {\bX}$. Let $\bu_1,\bu_2,\dots, \bu_n$ be the eigenvectors of $\widetilde \bSigma  = n^{-1} {\bX} {\bX}' $. It is well known that for $i=1,2,\dots,n$,
\beq
\label{VecRel}
\hat\bxi_i  = (n \hat \lambda_i)^{-1/2} \bX' \bu_i \;\; \text{and} \;\; \bu_i =  (n \hat \lambda_i)^{-1/2} \bX \hat\bxi_i \,,
\eeq

Using (\ref{VecRel}), we have
\beq
\label{eqB.2}
\hat\bxi_{jA} = \frac{\bLambda_A^{1/2} \bZ_A' \bu_j}{\sqrt{n \hat\lambda_j}} \;\; \text{and} \;\; \bu_j = \frac{{\bX} \hat\bxi_j}{\sqrt{n \hat\lambda_j}} = \frac{\bZ_A \bLambda_A^{1/2} \hat\bxi_{jA}}{\sqrt{n\hat\lambda_j}} + \frac{\bZ_B \bLambda_B^{1/2} \hat\bxi_{jB}}{\sqrt{n \hat \lambda_j}}\,.
\eeq
Since $\bu_j$ is the eigenvector of $\widetilde\bSigma$, that is, $n^{-1} {\bX}{\bX}' \bu_j = \hat \lambda_j \bu_j$. Plugging in $\bX = (\bZ_A \bLambda_A^{1/2}, \bZ_B \bLambda_B^{1/2})$, we obtain
\begin{equation*}
\Big( \bI_n - \frac{1}{n} \bZ_A \frac{\bLambda_A}{\lambda_j} \bZ_A' \Big) \bu_j = \bD \bu_j - \Delta \bu_j\,,
\end{equation*}
where we denote $\bD = (n\lambda_j)^{-1} \bZ_B\bLambda_B \bZ_B'$, $\Delta = \hat \lambda_j/\lambda_j - 1$.
We then left-multiply the above equation by $\bLambda_A^{1/2} \bZ_A'/\sqrt{n \hat\lambda_j}$ and employ the relationship (\ref{eqB.2}) to replace $\bu_j$ by $\hat \bxi_{jA}$ and $\hat\bxi_{jB}$ as follows:
\begin{equation}
\label{eqB.3}
\begin{aligned}
\Big(\bI_m - \frac{\bLambda_A}{\lambda_j} \Big) \hat\bxi_{jA} = & \frac{\bLambda_A^{1/2} (\frac 1 n \bZ_A' \bZ_A - \bI_m ) \bLambda_A^{1/2} }{\lambda_j} \hat\bxi_{jA} + \frac{\bLambda_A^{1/2} \bZ_A' \bD \bZ_A \bLambda_A^{1/2}}{n \hat\lambda_j} \hat\bxi_{jA} \\ & + \frac{\bLambda_A^{1/2} \bZ_A' \bD \bZ_B \bLambda_B^{1/2} }{n \hat\lambda_j} \hat\bxi_{jB} - \Delta \hat\bxi_{jA}\,.
\end{aligned}
\end{equation}
Further, we  define
\[
\bR = \sum_{k \in [m] \setminus j} \frac{\lambda_j}{\lambda_j - \lambda_k} \be_{kA} \be_{kA}'\,,
\]
where $\bR$ is well defined because $m > 1$.
Then we have $\bR (\bI_m - \bLambda_A/\lambda_j ) = \bI_m - \be_{jA} \be_{jA}'$. Left multiplying $\bR$ to (\ref{eqB.3}), we have
\begin{equation*}
\begin{aligned}
\hat\bxi_{jA}- \langle \hat\bxi_{jA}, \be_{jA} \rangle \be_{jA} = & \bR \Big(\frac{\bLambda_A}{\lambda_j}\Big)^{\frac12} \bK \Big(\frac{\bLambda_A}{\lambda_j}\Big)^{\frac12} \hat\bxi_{jA} \\
& + \bR \frac{\bLambda_A^{1/2} \bZ_A' \bD \bZ_B \bLambda_B^{1/2} }{n \hat\lambda_j} \hat\bxi_{jB} - \Delta \bR \hat\bxi_{jA}\,,
\end{aligned}
\end{equation*}
where $\bK = n^{-1} \bZ_A' \bZ_A - \bI_m + \lambda_j (n\hat \lambda_j)^{-1} \bZ_A' \bD \bZ_A $. Dividing both sides by $\|\hat\bxi_{jA}\|$, we get
\begin{equation}
\label{mainExpression_u}
\frac{\hat\bxi_{jA}}{\|\hat\bxi_{jA}\|} - \be_{jA} = \bR \Big(\frac{\bLambda_A}{\lambda_j}\Big)^{\frac12} \bK \Big(\frac{\bLambda_A}{\lambda_j}\Big)^{\frac12} \be_{jA} + \br_n\,,
\end{equation}
where
\beq
 \label{mainExpression_r}
  \begin{aligned}
\br_n =  & \Big(\langle \frac{\hat\bxi_{jA}}{\|\hat\bxi_{jA}\|}, \be_{jA} \rangle - 1 \Big) \be_{jA} + \bR \Big(\frac{\bLambda_A}{\lambda_j}\Big)^{\frac12} \bK \Big(\frac{\bLambda_A}{\lambda_j}\Big)^{\frac12} \Big( \frac{\hat\bxi_{jA}}{\|\hat\bxi_{jA}\|} - \be_{jA} \Big)  \\
& + \bR \frac{\bLambda_A^{1/2} \bZ_A' \bD \bZ_B \bLambda_B^{1/2}}{n \hat\lambda_j} \frac{\hat\bxi_{jB}}{\|\hat\bxi_{jA}\|} - \Delta \bR \Big( \frac{\hat\bxi_{jA}}{\|\hat\bxi_{jA}\|} - \be_{jA} \Big)\,.
  \end{aligned}
\eeq

Following \cite{FanWan15}, together with Lemma \ref{Key}, we can show $\|\br_n \| = o_P(n^{-\frac12})$. Further note that $(\bLambda_A/\lambda_j)^{\frac12} \be_{jA} = \be_{jA}$, we obtain
\beq \label{eqB.6}
\sqrt{n} \Big(\frac{\hat\bxi_{jA}}{\|\hat\bxi_{jA}\|} - \be_{jA}\Big)  = \sqrt{n} \bR \Big(\frac{\Lambda_A}{\lambda_j}\Big)^{\frac12} \bK \be_{jA} + o_P(1)\,.
\eeq
According to the definition of $\bR$, as $p \to \infty$,
\begin{equation*}
\bR \Big(\frac{\bLambda_A}{\lambda_j}\Big)^{\frac12} = \sum \limits_{k \in [m] \setminus j} \frac{\sqrt{\lambda_j \lambda_k}}{\lambda_j - \lambda_k} \be_{kA} \be_{kA}'  \to \sum \limits_{k \in [m] \setminus j} a_{jk} \be_{kA} \be_{kA}'  \,,
\end{equation*}
where $a_{jk} = \lim_{\lambda_j, \lambda_k \to\infty} \sqrt{\lambda_j \lambda_k}/(\lambda_j - \lambda_k)$ exists. So $\|\bR (\bLambda_A/\lambda_j)^{\frac12}\| = O(1)$. We claim $\|\bK\| = O_P(n^{-1/2})$, so the right hand side of (\ref{eqB.6}) is $O_P(1)$. To prove the rate of $\|\bK\|$, note first by Lemma \ref{Key}, $\| (p-m)^{-1} \bZ_B \bLambda_B \bZ_B' - \bar c \bI \| = O_P(\sqrt{n/p})$, so we have
\begin{align*}
\| \bD \| & = \Big\| \frac{1}{n} \frac{\bZ_B \bLambda_B \bZ_B'}{\lambda_j} \Big\| = O_P\Big(\frac{p-m}{n \lambda_j}\Big) = o_P({n^{-\frac12}})\,, \\
\|\bK\| &= \left\| \frac 1 n \bZ_A' \bZ_A - \bI_m + \frac{\lambda_j}{\hat \lambda_j} \frac 1 n \bZ_A' \bD \bZ_A \right\|  \\
& \le \left\| \frac 1 n \bZ_A' \bZ_A - \bI_m \right\| + \Big|\frac{\lambda_j}{\hat \lambda_j}\Big| \| \bD \| \left\| \frac 1 n \bZ_A' \bZ_A \right\| = O_P(n^{-\frac12})\,.
\end{align*}
Therefore, $\|\hat\bxi_{jA}/\|\hat\bxi_{jA}\| - \be_{jA}\| = O_P(n^{-1/2})$. From Lemma \ref{lemB.1}, $\|\hat\bxi_{jA}\| = 1 + O_P(1/n + 1/p)$, so (i) holds.

(ii) Now we prove the second conclusion of the theorem on the non-spiked part $\hat\bxi_{jB}$. Again we consider the case $m=1$ and $m > 1$ separately. When $m = 1$, by definition of eigenvector, we can easily see that
\[
\Big(\frac{\hat\lambda_1}{\lambda_1} - \frac{1}{n} \|\bz_1\|^2\Big) \hat\xi_{1A} = \frac{1}{n\sqrt{\lambda_1}} \bz_1'\bX_B \hat\bxi_{1B} \,;
\]
\[
\frac{\sqrt{\lambda_1}}{n} \bX_B'\bz_1 \hat\xi_{1A} + \frac{1}{n} \bX_B'\bX_B \hat\bxi_{1B} = \hat\lambda_1 \hat\bxi_{1B}\,.
\]
Plug the former equation into the latter one, we have
\[
\hat\lambda_1 \hat\bxi_{1B} = \bar\Delta^{-1} \frac{1}{n^2} \bX_B'\bz_1 \bz_1'\bX_B \hat\bxi_{1B} + \frac{1}{n} \bX_B'\bX_B \hat\bxi_{1B}\,,
\]
where $\bar\Delta = \hat\lambda_1/\lambda_1 - \|\bz_1\|^2/n$.
Therefore, let $\bOmega' = (\bgamma_1, \dots, \bgamma_{p})$, $\|\bOmega \hat\bxi_{1B}\|_{\max} = \max_{k \le p} |\bgamma_k' \hat\bxi_{1B}|$ is bounded by 
\[
\max_k \frac{|\bar\Delta^{-1}|}{\hat\lambda_1} \Big\|\frac1n \bgamma_k' \bX_B' \bz_1\Big\| \Big\| \frac1n \bz_1' \bX_B\Big\| \|\hat\bxi_{1B}\| + \Big\|\frac1n \bgamma_k'\bX_B' \bX_B \Big\| \|\hat\bxi_{1B}\| \,,
\]
which is $O_P(\sqrt{1/(n^2 p)}) = O_P(\sqrt{\log p/(np)})$ due to the facts that $\bar\Delta = O_P(n^{-1/2})$,  $\hat\lambda_1 = O_P(p)$, $\|\hat\bxi_{1B}\| = O_P(n^{-1/2})$ according to Lemma \ref{lemB.1} and three claims yet to be shown: 
\begin{equation} \label{EqB.7}
\begin{aligned}
\|\bgamma_k' \bX_B' \bZ_A / n\| = O_P(\sqrt{1/n}),\;\; & \| \bgamma_k'\bX_B' \bX_B /n\| = O_P(\sqrt{p/n}), \\
 \| \bZ_A' \bX_B / n\| =& O_P(\sqrt{p/n}) \,.
\end{aligned}
\end{equation}

Now we turn to the case $m > 1$. Similar as derivations above, by definition, we have
\[
\Big(\hat\lambda_j \bI_m - \frac1n \bLambda_A^{1/2} \bZ_A' \bZ_A \bLambda_A^{1/2}\Big) \hat\bxi_{jA} = \frac1n \bLambda_A^{1/2} \bZ_A' \bX_B \hat\bxi_{jB}\,;
\]
\[
\frac{1}{n} \bX_B'\bZ_A \bLambda_A^{1/2} \hat\bxi_{jA} + \frac{1}{n} \bX_B'\bX_B \hat\bxi_{jB} = \hat\lambda_j \hat\bxi_{jB}\,.
\]
The former equation implies
\[
\Big(\bI_m - \frac{\bLambda_A}{\lambda_j }\Big)  \hat\bxi_{jA} = \frac{1}{n\lambda_j } \bLambda_A^{1/2} \bZ_A' \bX_B \hat\bxi_{jB} - \bar\Delta \hat\bxi_{jA} \,,
\]
where $\bar\Delta = (\hat\lambda_j /\lambda_j - 1)\bI_m - (\bLambda_A/\lambda_j)^{1/2}(n^{-1}\bZ_A' \bZ_A - \bI_m) (\bLambda_A/\lambda_j)^{1/2}$. Note that this definition of $\bar\Delta$ degenerates to $\hat\lambda_1/\lambda_1 - \|\bz_1\|^2/n$ when $m = 1$ and $\|\bar\Delta\| = O_P(n^{-1/2})$. Left multiply this equation by $\bR$ defined in (i) and plug it into the previous equation, we obtain
\begin{align*}
\hat\lambda_j \bgamma_k' \hat\bxi_{jB} = & \frac{1}{n^2 \lambda_j} \bgamma_k' \bX_B'\bZ_A \bLambda_A^{1/2} \bR \bLambda_A^{1/2} \bZ_A'\bX_B \hat\bxi_{jB} + \frac{1}{n} \bgamma_k' \bX_B'\bX_B \hat\bxi_{jB} \\
& + \frac{1}{n} \langle \hat\bxi_{jA}, \be_{jA}\rangle  \bgamma_k'  \bX_B'\bZ_A \bLambda_A^{1/2} \be_{jA} - \frac{1}{n} \bgamma_k' \bX_B'\bZ_A \bLambda_A^{1/2} \bR \bar\Delta \hat\bxi_{jA} \,.
\end{align*}
Carefully bounding each term of the right hand side by (\ref{EqB.7}) and Lemma \ref{lemB.1}, we find the dominating term is the third term, which has rate $O_p(\sqrt{p/n})$. Thus $\|\bOmega \hat\bxi_{1B}\|_{\max} = \max_{k \le p} |\bgamma_k' \hat\bxi_{1B}| = O_P(\sqrt{1/(np)}) = O_P(\sqrt{\log p/(np)})$. 

It remains to prove (\ref{EqB.7}). Firstly,  $\| \bZ_A' \bX_B / n\|^2 \le \|\bZ_A ' \bZ_A/n\| \|\bZ_B \bLambda_B\bZ_B'/n \| = O_P(p/n)$. Thus the third result holds. To show the other two rates, denote $\bv = \bX_B \bgamma_k$. Then each element of $\bv$ is iid sub-Gaussian with bounded variance proxy. Hence $\bv'\bv / n = O_P(1)$ and $\bv'\bz_j /n = O_P(n^{-1/2})$ for $j \le m$. So we have
\[
\|\bgamma_k' \bX_B' \bZ_A / n\| \le \sqrt{m} \max_{j\le m} |\bv'\bz_j /n| = O_P(n^{-1/2}) \,,
\]
and 
\[
\| \bgamma_k'\bX_B' \bX_B /n\|^2 \le \|\bZ_B \bLambda_B\bZ_B'/n\| |\bv'\bv/n| = O_P(p/n) \,.
\]
Now the proof is complete.

\end{proof}

\begin{lem}
\label{lemB.1}
For $j \le m$, $\|\hat\bxi_{jA}\| = 1 + O_P(n^{-1} + p^{-1})$ and $\|\hat\bxi_{jB}\| = O_P(n^{-1/2} + p^{-1/2})$.
\end{lem}

\begin{proof}
Recall that $\bX = (\bZ_A \bLambda_A^{1/2}, \bZ_B \bLambda_B^{1/2})$. Let $\bZ = (\bZ_A, \bZ_B)$, then
\[
\bZ = \bX \bLambda_p^{-\frac 1 2} = \sqrt{n} \hat\bLambda_n^{\frac 12} (\hat\bxi_1, \dots, \hat\bxi_n)' \bLambda_p^{-\frac 1 2}\,,
\]
where $\bLambda_p = \diag(\bLambda_A, \bLambda_B)$ and $\hat\bLambda_n = \diag(\hat\lambda_1, \dots, \hat\lambda_n)$.  Further define $\bar \bLambda_p = \diag(1,\dots, 1, \lambda_{m+1},\dots,\lambda_{p})$ and consider the eigenvalue of the matrix $n^{-1} \bZ \bar\bLambda_p \bZ'$. The $j^{th}$ diagonal element of the matrix must lie in between its minimum and maximum eigenvalues. That is
\[
\lambda_n(\frac 1 n \bZ \bar\bLambda_p \bZ') \le \Big(\frac 1n \bZ \bar\bLambda_p \bZ'\Big)_{jj} = \hat\lambda_j \sum_{k=1}^p \hat\xi_{jk}^2 \frac{\bar\lambda_k}{\lambda_k} \le \lambda_1(\frac 1 n \bZ \bar\bLambda_p \bZ')\,,
\]
where $\hat\xi_{jk}$ is the $k$-th element of the $j^{th}$ empirical eigenvector for $j \le m$. Note that $\hat\lambda_j/\lambda_j$ converges to $1$, and decided by $\lambda_j$ both the left and right hand side converge in probability to $(m+\sum_{j=m+1}^p \lambda_j)/(n\lambda_j)$ by Lemma \ref{Key}, thus to $\bar c p /(n\lambda_j)$. So $\sum_{k=1}^p \hat\xi_{jk}^2 \bar\lambda_k/\lambda_k \overset{P} = O_P(1/n)$. Also, by definition, $\bar\lambda_k/\lambda_k = O(1/p)$ for $k \le m$ while the ratio is 1 for $k > m$. Hence, $\sum_{k=m+1}^p \hat\xi_{jk}^2 = O_P(1/n+ 1/p)$, which implies that $\|\hat\bxi_{jA} \| = \sqrt{1 - \sum_{k=m+1}^p \hat\xi_{jk}^2}  = 1+ O_P(n^{-1} + p^{-1})$ and $\|\hat\bxi_{jB}\| = O_P(n^{-1/2} + p^{-1/2})$.
\end{proof}

\section{Proofs in Section 4} \label{secC}

Let us introduce some additional notations and two lemmas in order to prove Theorem \ref{thm_eigenvectorHT}.
Assume $n$ is even, otherwise we can drop one sample without affecting the asymptotics. Let $\bar n = n/2$. For any permutation $\sigma$ of $\{1,2,\dots, n\}$, let $(i_1, i_2, \dots, i_n) := \sigma(1,2,\dots, n)$. For $r = 1, \dots, \bar n$, we define $\bw_r^{\sigma}$ and $\hat\bK_\sigma$:
\beq
\bw_r^{\sigma} = \bLambda_p^{\frac12} \bg_r / ({\bg_r}' \bLambda_p \bg_r)^{\frac12} \;\; \text{so that} \;\; k(\bX_{i_{2r-1}}, \bX_{i_{2r}}) \overset{d} = \bw_r^{\sigma} {\bw_r^{\sigma}}'\,,
\eeq
\beq
\hat\bK_\sigma = \frac {1}{\bar n} \sum_{r = 1}^{\bar n} \bw_r^{\sigma} {\bw_r^{\sigma}}'  \;\; \text{so that} \;\; \hat\bK \overset{d} = \frac{1}{\card(\mathcal S_n)} \sum_{\sigma \in \mathcal S_n} \hat\bK_{\sigma} \,,
\eeq
where $\mathcal S_n$ is the permutation group of $\{1,2,\dots, n\}$. For each fixed permutation $\sigma$, we have the following two conclusions on empirical eigenvalues and eigenvectors of $\hat\bK_\sigma$ similar to the sample covariance for sub-Gaussian factor models.

\begin{lem}
\label{lem_eigenvalueHT}
Under Assumptions \ref{assump1} and \ref{assump2}, for $j \le m$, we have
\begin{equation}
    | \hat \lambda_j^{\sigma} - \theta_j| = O_P(n^{-1/2})\,,
\end{equation}
and for $j > m$, $\hat\lambda_j^{\sigma} = O_P(n^{-1})$ where $\{\hat\lambda_j^{\sigma}\}$'s are eigenvalues of $\hat\bK_{\sigma}$.
\end{lem}

Now consider the leading empirical eigenvectors $\hat \bxi_j^{\sigma}$ of $\hat\bK_{\sigma}$, $j \le m$. Each $\hat\bxi_j^{\sigma}$ is divided into two parts $ \hat\bxi_j^{\sigma} = ( ({\hat\bxi_{jA}^{\sigma}})',  ({\hat\bxi_{jB}^{\sigma }})')'$ where $ \hat\bxi_{jA}^{\sigma}$ is of length $m$.

\begin{lem}
\label{lem_eigenvectorHT}
Under Assumptions \ref{assump1} and \ref{assump2}, for $j \le m$, we have \\
\begin{equation}
    \|\hat\bxi_{jA}^{\sigma} - \be_{jA} \|  = O_P(n^{-1/2})\,.
\end{equation}
In addition, $\|\hat \bxi_{jB}^{\sigma}\| = O_p(n^{-1/2} + p^{-1/2})$.
\end{lem}

With the above two lemmas, we prove Theorem \ref{thm_eigenvectorHT}.

\begin{proof} [\bf Proof of Theorem \ref{thm_eigenvectorHT}]
(i) First, we have the simple fact that
\[
\|\hat\bK - \bK\| = \Big\| \frac{1}{\card(\mathcal S_n)} \sum_{\sigma \in \mathcal S_n} \hat\bK_{\sigma} - \bK \Big\| \le \frac{1}{\card(\mathcal S_n)} \sum_{\sigma \in \mathcal S_n}  \|\hat\bK_{\sigma} -\bK\|\,.
\]
Now let us derive the rate of $\|\hat\bK_{\sigma} -\bK\|$. Write
\[
\hat\bK_{\sigma} - \bK = (\hat\bGamma_1, \hat\bGamma_2) \diag(\hat\bTheta_A, \hat\bTheta_B) (\hat\bGamma_1, \hat\bGamma_2)'  - \diag(\bTheta_A, \bTheta_B)\,,
\]
where $\bTheta_A =\diag(\theta_1, \dots, \theta_m)$ and $\bTheta_B = \diag(\theta_{m+1}, \dots,\theta_p)$. From Lemma \ref{lem_eigenvalueHT}, we have $\|\hat\bTheta_A - \bTheta_A \| = O_P(n^{-1/2})$ and $\|\hat\bTheta_B\| = O_P(n^{-1})$. From Lemma \ref{lem_eigenvectorHT}, we have $\|\hat\bGamma_1 - (\bI_m, {\bf 0})' \| = O_P(n^{-1/2})$. Therefore, the following two bounds hold:
\begin{align*}
\|\hat\bGamma_1\hat\bTheta_A \hat\bGamma_1' - \diag(\bTheta_A, {\bf 0}) \| \le & \|\hat\bGamma_1 (\hat\bTheta_A - \bTheta_A) \hat\bGamma_1' \| \\
& + \|(\hat\bGamma_1 - (\bI_m, {\bf 0})' ) \bTheta_A (\hat\bGamma_1 - (\bI_m, {\bf 0})' )' \| \\
& + \|(\bI_m, {\bf 0})' \bTheta_A (\hat\bGamma_1 - (\bI_m, {\bf 0})' )'\| \\
& + \|(\hat\bGamma_1 - (\bI_m, {\bf 0})' ) \bTheta_A  (\bI_m, {\bf 0})\| \,,
\end{align*}
which is $O_P(n^{-1/2})$; in addition,
\[
\|\hat\bGamma_2\hat\bTheta_B \hat\bGamma_2' - \diag({\bf 0}, \bTheta_B)\| \le \|\hat\bGamma_2\|^2 \|\hat\bTheta_B \| +  \|\bTheta_B\| = O_P(1/n + 1/p)\,,
\]
where $\|\bTheta_B\| = O(1/p)$. Therefore, $\|\hat\bK_{\sigma} - \bK\| = O_P(n^{-1/2})$ for any fixed permutation $\sigma$, which implies that $\|\hat\bK - \bK\|= O_P(n^{-1/2})$. This further implies conclusion (i) by Weyl's inequality and $\|\hat\bxi_{jB}\| = O_P(n^{-1/2})$.

(ii) We first prove the following conclusion: there exists diagonal scaling random matrix $\bD_j$ and random vector $\bh$ such that $\|\bD_j \hat \bxi_{jB} - {\bh}\| = O_P(\sqrt{\log n/(n p)})$ where $\bh$ is uniformly distributed over the centered sphere of dimension $p-m$ and radius $r = O_P(n^{-1/2})$.

To this end, we need to employ rescaled data $\bx_i^{R} =  \diag(\bI_m, \bOmega_0) \bx_i$ where $\bOmega_0 = \diag(\sqrt{\bar c/ \lambda_{m+1}}, \dots, \sqrt{\bar c/\lambda_p})$. Here superscript $R$ denotes rescaled data by $\diag(\bI_m, \bOmega_0)$. Recall that $\bx_i$ follows $EC({\bf 0}, \bLambda_p, \zeta)$.  After rescaling, $\bx_i^{R}$ follows $EC({\bf 0},  \diag(\bLambda_A, \bar c\bI_{p-m}), \zeta)$. Let $N = n(n-1)/2$ and define a $N \times p$ matrix $\bcalX$ such that $\bcalX_{((ii'), \cdot)} = (\bx_i - \bx_{i'})' / \|\bx_i - \bx_{i'}\|$, which corresponds to each pair of samples. Clearly $\hat\bK = N^{-1} \bcalX' \bcalX$. Let $\bcalX = (\bcalX_A, \bcalX_B)$, and correspondingly $\bcalX^{R} = \bcalL \bcalX \diag(\bI_m, \bOmega_0) = (\bcalL\bcalX_A, \bcalL\bcalX_B\bOmega_0)$ so that $\bcalX_A^{R} = \bcalL\bcalX_A$ and $\bcalX_B^{R} = \bcalL\bcalX_B\bOmega_0$, where
\[
\bcalL = \diag\bigl(\|\bx_i -\bx_{i'}\| / \|\bx_i^{R} -\bx_{i'}^{R}\|\bigr)\,.
\]
Other quantities are also defined for the rescaled data. For example, $\hat \bxi_j^{R}$ and $\bu_j^{R}$ are eigenvectors of $\hat\bK^{R}$ and $\widetilde\bK^{R} = N^{-1} \bcalX^{R} {\bcalX^{R}}'$. Let $\hat\bxi_j^{R} = ({\hat \bxi_{jA}^{R}}, {\hat\bxi_{jB}^{R}})'$. 

Since the estimator $\hat\bK^R$ is invariant to orthogonal transformation of the data, similar to \cite{Pau07}, we can show $\hat\bxi_{jB}^{R}/\|\hat\bxi_{jB}^{R}\|$ is distributed uniformly over the unit sphere. Define $\bh = \hat\bxi_{jB}^{R}$. From the proof of (i), we know $\|\bh\| = O_P(n^{-1/2})$. So $\bh$ is uniformly distributed over a centered ball of radius $O_P(n^{-1/2})$. Hence it only remains to bound the difference of $\hat\bxi_{jB}$ and $\bh$ to validate the claim.

Note that
\[
\|\widetilde\bK - \widetilde \bK^{R}\| \le \|N^{-1} (\bcalX_A \bcalX_A' - \bcalL \bcalX_A \bcalX_A' \bcalL) \| + \|N^{-1} (\bcalX_B \bcalX_B  - \bcalL \bcalX_B \bOmega_0^2 \bcalX_B' \bcalL)\|\,.
\]
It is not hard to show $\|\bcalL - \bI_{N}\| = O_p(\sqrt{\log n/p})$, which is in the same order as the first term. The second term is dominated by $O_P(\|\bcalL - \bI_{N}\|)$ plus
\begin{align*}
\|N^{-1} \bcalX_B (\bI_{p-m} - \bOmega_0^2) \bcalX_B' \| & = \|N^{-1} (\bI - \bOmega_0^2)^{\frac12} \bcalX_B'  \bcalX_B (\bI - \bOmega_0^2)^{\frac12} \| \\
\le \frac{1}{\card(\calS_n) } & \sum_{\sigma \in \calS_n}  \Big\|\frac{1}{\bar n} \sum_{r = 1}^{\bar n} (\bI - \bOmega_0^2)^{\frac12} \bw_{rB}^{\sigma} {\bw_{rB}^{\sigma}}' (\bI - \bOmega_0^2)^{\frac12} \Big\|\,,
\end{align*}
where $\bw_{rB}^{\sigma} = \bLambda_B^{1/2} \bg_{rB} / ({\bg_r}' \bLambda_p \bg_r)^{\frac12} $. Using the notations defined in the proof of Lemma \ref{lem_eigenvalueHT}, we have
\[
\Big\|\frac{1}{\bar n} \sum_{r = 1}^{\bar n} (\bI - \bOmega_0^2)^{\frac12} \bw_{rB}^{\sigma} {\bw_{rB}^{\sigma}}' (\bI - \bOmega_0^2)^{\frac12} \Big\| = \Big\|\frac{1}{\bar n} \bL \bR_B \bLambda_B^{1/2}(\bI - \bOmega_0^2) \bLambda_B^{1/2}\bR_B' \bL \Big\|\,,
\]
where $(\bR_B)_{i\cdot} = \bg_{iB} /\sqrt{p}$ and $$\bL = \diag((p^{-1}{\bg_1}' \bLambda_p \bg_1)^{-\frac12}, \dots, (p^{-1}{\bg_{\bar n}}' \bLambda_p \bg_{\bar n})^{-\frac12} ).$$
The right hand side is of order $O_P(\bar n^{-1} \sqrt{n/p}) = O_P(\sqrt{1/(np)})$ by Lemma \ref{Key}. This implies $\|\widetilde\bK - \widetilde \bK^{R}\| = O_P(\sqrt{\log n/p})$.
Thus by the $\sin \theta$ theorem of \cite{DavKah70}, we get $\| \bu_j - \bu_j^{R}\| = O_P( \sqrt{\log n/p})$. With (\ref{VecRel}), we have
\[
\|\sqrt{\hat\lambda_j/\hat\lambda_j^R}\bOmega_0 \hat\bxi_{jB} - \bh\| = \Bigg\|\frac{\bOmega_0 \bcalX_B' \bu_j}{\sqrt{N\hat\lambda_j^R}} - \frac{{{\bcalX_B^{R}}'} \bu_j^{R}}{\sqrt{N\hat\lambda_j^{R}}} \Bigg\| \le \|\bOmega_0\| \Bigg\|\frac{ \bcalX_B' }{\sqrt{N \hat\lambda_j^R}}\Bigg\| \| \bu_j- \bcalL \bu_j^{R} \|  \,.
\]
The right hand side is $O_P( \sqrt{\log n/(np)})$ since from above $\| \bu_j - \bu_j^{R}\| = O_P( \sqrt{\log n/p})$, $\|\bcalL - \bI\|=O_P(\sqrt{\log n/p})$, $\hat\lambda_j^{R}$ converges to $\theta_j$ and $\|\bcalX_B'/\sqrt{N}\| = O_P(n^{-1/2})$, which is true because
\[
\Big\|N^{-1} \bcalX_B' \bcalX_B \Big\| \le \frac{1}{\card(\calS_n) } \sum_{\sigma \in \calS_n} \Big\|\frac{1}{\bar n} \sum_{r = 1}^{\bar n} \bw_{rB}^{\sigma} {\bw_{rB}^{\sigma}}' \Big\| = O_P(1/n)\,,
\]
where the last equality can be seen from the proof of Lemma \ref{lem_eigenvalueHT}.
Hence we conclude, if $\bD_j = \sqrt{\hat\lambda_j/\hat\lambda_j^R}\bOmega_0$,
\[
  \|\bD_j \hat\bxi_{jB}  - \bh\| =  O_P( \sqrt{\log n/(np)})\,.
\]
We are done proving the claim. 

Now let us come back to our goal of bounding $\|\bOmega \hat\bxi_{jB}\|_{\max}$ for any $p$ by $p-m$ matrix $\bOmega$ such that $\bOmega'\bOmega = \bI_{p-m}$. Obviously,
\[
\| \bOmega \hat\bxi_{jB} \|_{\max} \le \| \bOmega \bD_j^{-1}(\bD_j\hat\bxi_{jB} - \bh) \|_{\max} + \|\bOmega\bD_j^{-1}\bh\|_{\max} = I + II \,.
\]
Let $\bOmega' = (\bgamma_1, \dots, \bgamma_p)$. Then,
\[
I \le \max_{k\le p} \|\bD_j^{-1} \bgamma_k\| \|\bD_j\hat\bxi_{jB} - \bh\| \le O_P(1) \|\bD_j\hat\bxi_{jB} - \bh\| = O_P(\sqrt{\log n/(np)})\,,
\]
where the second inequality is due to $\lambda_{k} \le C_0$ for $k > m$, $\|\bgamma_k\| \le 1$ and $\hat\lambda_j /\hat\lambda_j^R = 1 + o_P(1)$.
Thus to bound the elementwise sup-norm $\| \bOmega \hat\bxi_{jB} \|_{\max}$, it suffices to show $II = O_P(\sqrt{\log p/(np)})$.

Let $\bG = (G_1, \dots, G_{p-m})$ be standard normal distributed. Obviously, $\bh \overset{d} = \|\bh\| \cdot \bG/\|\bG\|$ where $\bG/\|\bG\|$ is uniform over unit sphere of dimension $p-m$. Provided $\|\bh\| = O_P(n^{-1/2})$, we only need to show $\|\bOmega\bD_j^{-1} \bG\|_{\max} /\|\bG\| = O_P(\sqrt{\log p/p})$. It follows from $\hat\lambda_j /\hat\lambda_j^R = 1 + o_P(1)$ and
\[
\max_{k \le p} |\bgamma_k' \bOmega_0^{-1} \bG| /\|\bG\| = O_P( \sqrt{\log p/p})\,,
\]
since $\bgamma_j' \bOmega_0^{-1} \bG$ is normally distributed with bounded variance. This completes the proof for (ii).

\end{proof}

\begin{proof}[\bf Proof of Lemma \ref{lem_eigenvalueHT}]
Recall that
\[
\bw_r^{\sigma} = \bLambda_p^{\frac12} \bg_r / ({\bg_r}' \bLambda_p \bg_r)^{\frac12} \;\;\text{and}\;\;
\hat\bK_\sigma = \frac {1}{\bar n} \sum_{r = 1}^{\bar n} \bw_r^{\sigma} {\bw_r^{\sigma}}' \,.
\]
For ease of notation, let us assume $\sigma$ is the identity permutation and ignore the index $\sigma$ in the following.
Define $\bW = (\bw_1^{\sigma}, \dots, \bw_{\bar n}^{\sigma})' = (\widetilde\bw_1, \dots, \widetilde \bw_p) = (\bZ_A \bLambda_A^{1/2}, \bZ_B \bLambda_B^{1/2})$ where $\bZ_A = (\widetilde\bz_1,\dots, \widetilde\bz_m)$, $\bZ_B = (\widetilde\bz_{m+1},\dots,\widetilde\bz_p)$ with
\[
\widetilde\bz_j = \bL (g_{1j},\dots,g_{rj})'/\sqrt{p}
\]
and
\[
 \bL = \diag((p^{-1}{\bg_1}' \bLambda_p \bg_1)^{-\frac12}, \dots, (p^{-1}{\bg_{\bar n}}' \bLambda_p \bg_{\bar n})^{-\frac12} ) \,.
\]
Then, $\hat \bK_{\sigma} ={\bar n}^{-1} \bW' \bW$. Exchanging $\bW'$ and $\bW$, we further define $\widetilde \bK_{\sigma}  = {\bar n}^{-1} \bW \bW'$, which share the same nonzero eigenvalues as $\hat \bK_{\sigma}$. Now in order to prove the lemma, let us decompose $\widetilde \bK_{\sigma} = \bA + \bB$ where
\[
\bA = \bar n^{-1} \sum_{j=1}^m \lambda_j \widetilde\bz_j \widetilde\bz_j', \quad \mbox{and} \quad
\bB = \bar n^{-1} \sum_{j=m+1}^p  \lambda_j \widetilde\bz_j \widetilde\bz_j' \,.
\]

We deal with $\bA$ first. Note that $\bA$ has the same eigenvalues as matrix $\widetilde \bA = \bar n^{-1} \bLambda_A^{1/2} \bZ_A' \bZ_A \bLambda_A^{1/2}$, where $\bZ_A$ is an $\bar n \times m$ matrix with iid rows. Therefore, we are in the low dimensional situation with fixed dimension $m$. It is easy to see that
\[
(\widetilde\bA)_{i,j} = \frac{1}{\bar n} \sum_{r = 1}^{\bar n} \frac{g_{ri} g_{rj} \sqrt{\lambda_i\lambda_j}}{ \sum_{k=1}^p \lambda_k g_{rk}^2} \,,
\]
and thus by the central limit theorem, the $j^{th}$ diagonal element of $\widetilde \bA$ is $\theta_j+O_P(n^{-1/2})$ and the $(i,j)^{th}$ off-diagonal elements are of order $O_P(n^{-1/2})$. Therefore, write $\bK = \diag(\bTheta_A, \bTheta_B)$, we have $\tilde\bA = \bTheta_A + \bH$ with $\|\bH\| = O_P(n^{-1/2})$. By Weyl's inequality,
\[
|\lambda_j(\bA) - \theta_j| = |\lambda_j(\tilde\bA) - \lambda_j(\bTheta_A)| \le \|\bH\| = O_P(n^{-1/2})\,.
\]

By the definition of $\bB$, $\bB = \bar n^{-1} \bZ_B \bLambda_B \bZ_B' $, where the $i^{th}$ row of $\bZ_B$ is $(\bZ_B)_{i\cdot} = \bg_{iB}/\|\bLambda_p^{1/2}\bg_i\|$. Let $\bZ_B = \bL\bR_B $ where $(\bR_B)_{i\cdot} = \bg_{iB} /\sqrt{p}$. Therefore
\[
\bar n \lambda_{k}(\bB) = \lambda_k(\bZ_B \bLambda_B \bZ_B') \le\lambda_k( \bR_B \bLambda_B \bR_B') \lambda_{\max} (\bL^2)\,.
\]
Since each row of $\bR_B$ is Gaussian, by Lemma \ref{Key} with $t = \sqrt{\bar n}$, for any $k \le \bar n$,
\[
\lambda_k(\bR_B \bLambda_B \bR_B') = \frac{1}{p-m}\sum_{j=m+1}^p \lambda_j + O_P\Big(\sqrt{\frac{n}{p}}\Big) = \bar c + O_P\Big(\sqrt{\frac{n}{p}}\Big) + o(1)\,.
\]
In addition, we have $\lambda_{\max}(\bL^2) = O_P(1)$. This is because
\[
\lambda_{\max}(\bL^2) = \Big( \min_{i \le n} \frac1p \sum_{j=1}^p \lambda_j g_{ij}^2 \Big)^{-1}\!\!\!\! \le \Big( \min_{i \le n} \frac1p \sum_{j=m+1}^p \lambda_j g_{ij}^2 \Big)^{-1}\!\!\!\!
= \frac{1}{\bar c} + O_P(\sqrt{\log n/p})\,.
\]
Therefore, $\lambda_k(\bB) = O_P(n^{-1})$. By Wely's Theorem, $\lambda_j(\bA) + \lambda_n(\bB) \le \hat\lambda_j^{\sigma} \le \lambda_j(\bA) + \lambda_1(\bB).$
Therefore, we conclude that $|\hat \lambda_j^{\sigma} - \theta_j|= O_P(n^{-1/2})$ for $j \le m$ and $\hat\lambda_j^{\sigma} = O_P(n^{-1})$ for $j > m$.
\end{proof}

\begin{proof} [\bf Proof of Lemma \ref{lem_eigenvectorHT}]

Similar to Lemma \ref{lemB.1}, we can prove $\|\hat\bxi_{jA}^{\sigma}\| = 1 + O_P(1/n + 1/p)$ as well as $\|\hat\bxi_{jB}^{\sigma}\| = O_P(n^{-1/2} + p^{-1/2})$. If $m = 1$, obviously the conclusion holds. So in the following we assume $m > 1$.

Write $\bW = (\bZ_A \bLambda_A^{\frac{1}{2}}, \bZ_B \bLambda_B^{\frac{1}{2}})$. Define $\bu_j$ as the eigenvector of $\widetilde \bK_{\sigma} = {\bar n}^{-1} \bW\bW'$. We obtain
\begin{equation*}
\Big( \theta_j \bI_{\bar n} - \frac{1}{\bar n} \bZ_A \bLambda_A \bZ_A' \Big) \bu_j = \bD \bu_j - \Delta \bu_j\,,
\end{equation*}
where we denote $\bD = {\bar n}^{-1} \bZ_B\bLambda_B \bZ_B'$, $\Delta = \hat \lambda_j -\theta_j$.
We then left multiply the above equation by $\bLambda_A^{1/2} \bZ_A'/\sqrt{\bar n \hat\lambda_j}$ and employ the relationship (\ref{eqB.2}) to replace $\bu_j$ by $\hat \bxi_{jA}^{\sigma}$ and $\hat\bxi_{jB}^{\sigma}$ to obtain
\begin{equation}
\label{eqC.1}
\begin{aligned}
\Big(\theta_j \bI_m - \bTheta_A \Big) \hat\bxi_{jA}^{\sigma} = & \Big(\frac{1}{\bar n} \bLambda_A^{1/2} \bZ_A' \bZ_A \bLambda_A^{1/2} - \bTheta_A \Big) \hat\bxi_{jA}^{\sigma} + \frac{\bLambda_A^{1/2} \bZ_A' \bD \bZ_A \bLambda_A^{1/2}}{\bar n \hat\lambda_j} \hat\bxi_{jA}^{\sigma} \\ & + \frac{\bLambda_A^{1/2} \bZ_A' \bD \bZ_B \bLambda_B^{1/2} }{\bar n \hat\lambda_j} \hat\bxi_{jB}^{\sigma} - \Delta \hat\bxi_{jA}^{\sigma}\,.
\end{aligned}
\end{equation}
Further, we define
\[
\bR = \sum_{k \in [m] \setminus j} \frac{1}{\theta_j - \theta_k} \be_{kA} \be_{kA}'\,.
\]
Then we have $\bR (\bI_m - \bLambda_A/\lambda_j ) = \bI_m - \be_{jA} \be_{jA}'$. Left multiplying $\bR$ to (\ref{eqC.1}),
\begin{equation*}
\hat\bxi_{jA}^{\sigma}- \langle \hat\bxi_{jA}, \be_{jA} \rangle \be_{jA} = \bR  \bH  \hat\bxi_{jA}^{\sigma}  + \bR \frac{\bLambda_A^{1/2} \bZ_A' \bD \bZ_B \bLambda_B^{1/2} }{\bar n \hat\lambda_j} \hat\bxi_{jB}^{\sigma} - \Delta \bR \hat\bxi_{jA}^{\sigma}\,,
\end{equation*}
where $\bH = {\bar n}^{-1} \bLambda_A^{1/2} \bZ_A' \bZ_A \bLambda_A^{1/2} - \bTheta_A  + ({\bar n}\hat \lambda_j)^{-1} \bLambda_A^{1/2} \bZ_A' \bD \bZ_A \bLambda_A^{1/2}$. Dividing both sides by $\|\hat\bxi_{jA}^{\sigma}\|$, we get
\begin{equation}
\label{mainExpression_uHT}
\hat\bxi_{jA}^{\sigma}/\|\hat\bxi_{jA}^{\sigma}\| - \be_{jA} = \bR \bH \be_{jA} + \br_n\,,
\end{equation}
where
\beq
 \label{mainExpression_rHT}
  \begin{aligned}
\br_n =  & \Big(\bigl\langle \frac{\hat\bxi_{jA}^{\sigma}}{\|\hat\bxi_{jA}^{\sigma}\|}, \be_{jA} \bigr\rangle - 1 \Big) \be_{jA} + \bR \bH \Big( \frac{\hat\bxi_{jA}^{\sigma}}{\|\hat\bxi_{jA}^{\sigma}\|} - \be_{jA} \Big)  \\
& + \bR \frac{\bLambda_A^{1/2} \bZ_A' \bD \bZ_B \bLambda_B^{1/2}}{\bar n \hat\lambda_j} \frac{\hat\bxi_{jB}^{\sigma}}{\|\hat\bxi_{jA}^{\sigma}\|} - \Delta \bR \Big( \frac{\hat\bxi_{jA}^{\sigma}}{\|\hat\bxi_{jA}^{\sigma}\|} - \be_{jA} \Big)\,.
  \end{aligned}
\eeq

Following \cite{FanWan15}, we are able to show $\|\br_n \| = o_P(n^{-\frac12})$. In addition, from the proof of Lemma \ref{lem_eigenvalueHT}, we have $\| \bD \| = \|\bar n^{-1} \bZ_B \bLambda_B \bZ_B' \| = O_P(n^{-1}) = o_P({n^{-\frac12}})$. So
\begin{align*}
\|\bH\| &\le  \Bigl\| n^{-1} \bLambda_A^{1/2} \bZ_A' \bZ_A \bLambda_A^{1/2} - \bTheta_A \Bigr\|  + \frac{\| \bD \|}{\hat \lambda_j} \Bigl\| \frac 1 n\bLambda_A^{1/2} \bZ_A' \bZ_A \bLambda_A^{1/2}\Bigr\| = O_P(n^{-\frac12})\,.
\end{align*}
This gives $\|\hat\bxi_{jA}^{\sigma}/\|\hat\bxi_{jA}^{\sigma}\| - \be_{jA}\|  = O_P(n^{-1/2})$ since clearly $\|\bR\| = O(1)$. Together with $\|\hat\bxi_{jA}^{\sigma}\| = 1 + O_P(1/n + 1/p)$, it follows $\|\hat\bxi_{jA}^{\sigma} - \be_{jA} \|  = O_P(n^{-1/2})$. The proof is complete.
\end{proof}

\section{A Technical Lemma} \label{secD}

Recall that $\bz_{i} = \bLambda_p^{-1/2} \bx_i$ is the standardized version of the transformed data $\bx_i$. We have the following theorem for $\bz_i$, which will be useful for the proofs in Section \ref{sec3} and \ref{sec4}.

\begin{lem} \label{Key}
Let $\bZ$ be the $n\times p$ matrix ($p \ge n$) with rows $\bz_i'$. Assume $\bz_i$ to be iid sub-Gaussian random vector with $\|\bz_i\|_{\phi_2} = \sup_{\bu \in \mathcal S^{p-1}} |\langle \bz_i, \bu \rangle|_{\phi_2} \le M$ for some constant $M>0$. Condition (\ref{indApproxCond}) holds for $\bz_i$.
The columns of $\bZ$ are denoted by $\widetilde\bz_j$ of length $n$. Then for $\forall t \ge 0$, let $\delta = C_0\sqrt{\frac{n}{p}} + \frac{t}{\sqrt{p}}$, we have
\beq\label{eqD.1}
\max_{i \le n} \Big| \lambda_{i} \Big(\frac 1p \sum_{i=1}^p w_j \widetilde\bz_j \widetilde\bz_j' \Big) -\bar w \Big| \le \max\{ \delta^2, \delta\} \,,
\eeq
with probability at least $1 - 2 exp(-c_0 t^2)$, where $C_0, c_0 > 0$ depend on $M, M_1, M_2$. Here, $|w_j|$ is bounded from above for all $j$ and $\bar w = p^{-1}\sum_{j=1}^p w_j$.
\end{lem}

\begin{proof}
Without loss of generality, let us assume all the $w_j$'s are non-negative and bounded away from zero. Otherwise, we subtract the minimal one from all the $w_j$'s. Let the new non-negative weights to be $\widetilde w_j = w_j - w_{\min} + 1$. Since all the $n$ eigenvalues concentrate to the same number, it is easy to separately consider the concentration for $\lambda_{i} (p^{-1} \sum_{i=1}^p \widetilde w_j \widetilde\bz_j \widetilde\bz_j' )$ and $\lambda_{i} (p^{-1} \sum_{i=1}^p \widetilde\bz_j \widetilde\bz_j' )$, which both have nonnegative lower bounded weights.

Let $\bD = \diag(w_1, \dots, w_p)$, so $p^{-1} \sum_{i=1}^p w_j \widetilde\bz_j \widetilde\bz_j' = p^{-1} \bZ\bD\bZ'$. Assume without loss of generality that $w_j$ is decreasing and $w_p = 1$. First we have
\[
\lambda_{\max} \Big(\frac 1 p \bZ\bD\bZ'\Big) = \lambda_{\max} \Big(\frac 1 p \bD^{1/2}\bZ'\bZ \bD^{1/2}\Big) \le  \lambda_{\max}(\bD) \lambda_{\max}\Big(\frac 1 p \bZ'\bZ\Big) \,.
\]
Since $\bz_i$ is sub-Gaussian, by Theorem 5.39 of \cite{Ver10}, we have with probability at least $1-2\exp(-c_1 t^2)$,
\[
\lambda_{\max}\Big(\frac 1 n \bZ'\bZ\Big) \le \Big(1 + C_1\sqrt{\frac{p}{n}} + \frac{t}{\sqrt{n}}\Big)^2.
\]

If $t \ge C_1\sqrt{w_1 p}$,  without loss of generality we assume $C_1 \ge 1$. Then since $|\bar w| \le w_1 \le \delta^2$, the minimum eigenvalue of $p^{-1} \bZ\bD\bZ'$ satisfies the conclusion. It remains to validate the conclusion for the maximal eigenvalue. According to the above, $\lambda_{\max}(\bD) \lambda_{\max}(\bZ'\bZ/p)$ is bounded by
\begin{align*}
w_1 \frac{n}{p} \Big(1 + C_1\sqrt{\frac{p}{n}} & + \frac{t}{\sqrt{n}}\Big)^2 \le \frac{n}{p} \Big(\sqrt{w_1} + \sqrt{\bar w}\sqrt{\frac{p}{n}} + \frac{(1+\sqrt{w_1})t}{\sqrt{n}}\Big)^2 \\
& = \Big(\sqrt{\bar w} + \sqrt{w_1} \sqrt{\frac{n}{p}} + \frac{(1+\sqrt{w_1})t}{\sqrt{p}}\Big)^2 \le \bar w + \max\{\delta^2, \delta\}\,.
\end{align*}
Thus the theorem holds for $t \ge C_1\sqrt{w_1 p}$.

We only need to consider the case $t < C_1\sqrt{w_1 p}$. This corresponds to conditioning on the event $\mathcal E = \{ \lambda_{\max}(\bZ\bD\bZ'/p) \le C_2^2  \}$ where $C_2 \ge \sqrt{w_1}(1 + C_1 + C_1 \sqrt{w_1})$. Obviously, with probability at least $1-2\exp(-c_1 t^2)$, the event $\mathcal E$ holds. To prove (\ref{eqD.1}), it suffices to show that with high probability
\[
\Big\|\frac 1p \bZ\bD\bZ' - \bar w \bI \Big\| \le 2 \max_{\bx \in \mathcal N} \Big|\frac{1}{p} \|\bD^{\frac12} \bZ' \bx\|^2 - \bar w \Big| \le \delta\,,
\]
where $\mathcal N$ is the $\frac14$-net covering the unit sphere $\mathcal S^{n-1}$ and $|\mathcal N| \le 9^n$ \citep{Ver10}. Due to the following decomposition,
\[
\|\bD^{\frac12} \bZ' \bx\|^2 = \Big\|\sum_{i}^n x_i \bD^{\frac12} \bz_i \Big\|^2 = \sum_{j=1}^p w_j + \sum_{i=1}^n x_i^2 (\bz_i'\bD\bz_i - \tr(\bD)) + \sum_{j \ne k} x_j x_k \bz_j'\bD \bz_k\,,
\]
we have
\[
\Big|\frac{1}{p} \|\bD^{\frac12} \bZ' \bx\|^2 - \bar w \Big| \le \Big|\frac1p\sum_{i=1}^n x_i^2 (\bz_i'\bD\bz_i - \tr(\bD))\Big| + \Big|\frac1p\sum_{j \ne k} x_j x_k \bz_j'\bD\bz_k\Big| =: | \Delta_1| + |\Delta_2| \,.
\]
Therefore,
\begin{equation} \label{eqD.2}
\begin{aligned}
\mathbb P\Big(\Big\|\frac 1p \bZ\bD\bZ' - \bar w \bI \Big\|  > & \delta\Big) \le \mathbb P\Big(\max_{\bx \in \mathcal N} \Big|\frac{1}{p} \|\bD^{\frac12} \bZ' \bx\|^2 - \bar w \Big| > \delta/2 \Big) \\
& \le \mathbb P\Big( \max_{\bx \in \mathcal N} |\Delta_1| > \delta/4\Big) +  \mathbb P\Big( \max_{\bx \in \mathcal N} |\Delta_2| > \delta/4\Big)\,.
\end{aligned}
\end{equation}
We need to separately bound the two terms on the right hand side.

Since $|\Delta_1| \le \max_{i \le n} |\bz_i'\bD\bz_i - \tr(\bD)|/p$,
\beq \label{eqD.3}
\mathbb P\Big( \max_{\bx \in \mathcal N} |\Delta_1| > \delta/4\Big) \le |\mathcal N| \cdot n \cdot \max_{\bx \in \mathcal N, i \le n} \mathbb P\Big( |\bz_i'\bD\bz_i - \tr(\bD)| > \frac{\delta p}{4}  \Big)\,.
\eeq
For a fixed $\bx$ and $i$, we now bound $\mathbb P( |\bz_i'\bD\bz_i - \tr(\bD)| > \delta p/4 )$. By Lemma \ref{QuadForm}, choosing $\bA = \bD^{1/2} /\sqrt{p}$, since $\bz_i$ satisfies (\ref{indApproxCond}), we have
\[
\mathbb P\Big( |\bz_i'\bD\bz_i - \tr(\bD)|/p > C_{\psi} w_1 \Big(\sqrt{\frac u p} + \frac u p\Big) \Big) \le 3 \exp(-u)\,,
\]
where $C_{\psi}$ is defined in Lemma \ref{QuadForm} and is bounded since $\psi = w_1$ is bounded.
Choose $u = C_3\delta^2 p$ so that $u/p < 1$ and $C_3 < 1/(64C_{\psi}w_1)$,  which implies $\delta/4 > 8C_{\psi}w_1 \sqrt{u/p} > C_{\psi} w_1 (\sqrt{u/p} + u/p)$. So $\mathbb P( |\bz_i'\bD\bz_i - \tr(\bD)| > \delta p/4 )$ is bounded by $3 \exp(-C_3 \delta^2 p)$.
Therefore, from (\ref{eqD.3}), we have
\[
\mathbb P\Big( \max_{\bx \in \mathcal N} |\Delta_1| > \delta/4\Big) \le 9^n \cdot n \cdot 3\exp(-C_3\delta^2p)  \le \exp(-c_0 t^2)
\]
by choosing $C_0$ in the definition of $\delta$ large enough. This proves the first term in (\ref{eqD.2}).

For the second term  in (\ref{eqD.2}), we apply the decoupling technique. By Lemma 5.60 of \cite{Ver10},
\[
|\Delta_2| \le \frac{4}{p} \max_{\mathcal T \subseteq [n]} \Big| \sum_{j \in \mathcal T, k \in \mathcal T^c} x_j x_k \bz_j'\bD \bz_k\Big| \,.
\]
So we have
\beq \label{eqD.4}
\mathbb P\Big( \max_{\bx \in \mathcal N} |\Delta_2| > \delta/4\Big) \le |\mathcal N| |\mathcal T|\cdot \max_{\bx, \mathcal T} \mathbb P\Big(\Big| \sum_{j \in \mathcal T, k \in \mathcal T^c} x_j x_k \bz_j'\bD \bz_k\Big| > \frac{\delta p}{16}\Big)\,.
\eeq
For each fixed $\bx$ and $\mathcal T$, we first consider the above probability conditioning on $\bz_k$ for $k \in \mathcal T^c$. Let $H_j = \sum_{k \in \mathcal T^c} x_k \bz_j'\bD \bz_k = \bz_j' \bD \bZ_{\mathcal T^c} \bx_{T^c}$ where $\bZ_{\mathcal T^c}$ is constructed by columns $\bz_k$ for $k \in \mathcal T^c$ and $\bx_{\mathcal T^c}$ contains the coordinates of $\bx$ corresponding to $\mathcal T^c$ . We know $H_j$ is sub-Gaussian since
\[
\|H_j\|_{\phi_2} \le \|\bz_j\|_{\phi_2} \|\bD^{1/2}\| \|\bD^{1/2}\bZ_{\mathcal T^c}\| \|\bx_{T^c}\| \le  \sqrt{w_1 } \|\bz_j\|_{\phi_2} \|\bD^{1/2}\bZ\| \le C_2 M \sqrt{w_1 p}\,,
\]
where the last inequality is due to conditioning on the event $\mathcal E$. So there exists a constant $C_4 > 0$ independent of $\bZ_{\mathcal T^c}$ such that $\|H_j\|_{\phi_2} \le C_4 \sqrt{p}$. Furthermore, the weighted sum of $H_j$'s is also sub-Gaussian distributed with $\|\sum_{j \in \mathcal T} x_j H_j \|_{\phi_2} \le (\sum_{j \in \mathcal T} x_j^2 \|H_j\|_{\phi_2}^2)^{1/2} \le C_4 \sqrt{p}$. Hence, from (\ref{eqD.4}),
\[
\mathbb P\Big( \max_{\bx \in \mathcal N} |\Delta_2| > \delta/4\Big) \le 9^n \cdot 2^n \cdot \mathbb E\Big[ \mathcal P\Big(\Big| \sum_{j \in \mathcal T} x_j H_j \Big| > \frac{\delta p}{16} \Big| \bZ_{\mathcal T^c}\Big) \Big] \,,
\]
where the right hand side is bounded by
\[
9^n \cdot 2^n \cdot 2 e^{-C \delta^2 p^2 / (256 C_4^2 p)} \le \exp(-c_0 t^2)
\]
by choosing a large enough $C_0$ in the definition of $\delta$. So we bounded the second term.

To conclude, from (\ref{eqD.2}),
\[
\mathbb P\Big(\bigl\|\frac 1p \bZ\bD\bZ' - \bar w \bI\bigr\| > \delta\Big) \le 2\exp(-c_0 t^2) \,,
\]
which implies (\ref{eqD.1}).
\end{proof}

\begin{lem} \label{QuadForm}
Let $\bA$ be a $m$ by $n$ matrix and $\bSigma := \bA'\bA$. Suppose $\bx = (x_1, \dots, x_n)$ is an isotropic sub-Gaussian random vector, that is,
\[
\mathbb E[\exp(\balpha' \bx)] \le \exp(\|\balpha\|^2/2)\,,
\]
for all $\balpha \in \mathbb R^n$. For all $t > 0$,
\beq \label{eqD.5}
\mathbb P\Big( \|\bA \bx\|^2 > \tr(\bSigma) + 2\sqrt{\tr(\bSigma^2)t} + 2\|\bSigma\| t \Big) \le e^{-t}\,;
\eeq
if furthermore the SVD decomposition of $\bA = \bU \bD \bV'$ where $\bD$ is a $m$ by $m$ diagonal matrix and $\bU,\bV$ consist of left and right orthogonal singular vectors and (\ref{indApproxCond}) holds for $\bV' \bx$, we have,
\beq \label{eqD.6}
\mathbb P\Big( \|\bA \bx\|^2 < \tr(\bSigma) - C_{\psi} \sqrt{\tr(\bSigma^2)t} - C_{\psi} \|\bSigma\| t \Big) \le 2e^{-t}\,,
\eeq
where $C_{\psi} = \max\{2\psi+2\psi\sqrt{M_2}, 2+2M_1^{-1}\}$ and $\psi = \lambda_{1}(\bD^2)/\lambda_{m}(\bD^2)$ is the condition number of $\bSigma$.
\end{lem}

The above lemma is an extension of the exponential inequality for iid one dimensional sub-Gaussian variables proved by \cite{LauMas00}. It is the Hanson-Wright inequality for the quadratic functional of a sub-Gaussian random vector. \cite{RudVer13} showed this inequality for independent sub-Gaussian elements. \cite{HsuKakZha12} obtained the upper tail bound (\ref{eqD.5}) under a much weaker assumption of general sub-Gaussian vector with dependency. However, they did not provide result for the lower tail bound. Note that quadratic functionals are different from linear functionals in that changing the sign of $\bx$ does not naturally give the lower tail bound. In the following, we prove (\ref{eqD.6}) under (\ref{indApproxCond}). This bound is used for proving Lamma \ref{Key}.

\begin{proof}

Denote $\by = \bV' \bx \in \mathbb R^m$ so that $\|\bA\bx\|^2 = \by' \bD^2 \by$. Write $\bD^2 = \diag(d_1, \dots, d_m)$ with decreasing diagonal elements. Since (\ref{indApproxCond}) holds for $\by$, we have for $\theta \le M_1$,
\[
\mathbb P\Big( \by'\by < m - 2\sqrt{m M_2 t} - 2 M_1^{-1} t \Big) \le \exp\Big(-2\theta\sqrt{m M_2 t} - 2 \theta M_1^{-1} t + M_2 \theta^2 m\Big) \,.
\]
Choose $\theta = \sqrt{t/(M_2 m)}$ which is smaller than $M_1$ if $t \le m M_1^2 M_2$ while choose $\theta = M_1$ if $t > m M_1^2 M_2$. In any case, we can show that the right hand side is bounded by $\exp(-t)$. Define an event $\mathcal E = \{ \by'\by \ge m - 2\sqrt{m M_2 t} - 2 M_1^{-1} t \}$. Then $\mathbb P(\mathcal E^c) \le \exp(-t)$.  Futhermore, we define
\[
\mathcal A = \Big\{\by' \bD^2 \by <  \sum\nolimits_{j \le m} d_j - C_{\psi} \sqrt{t \sum\nolimits_{j \le m} d_j^2} - C_{\psi} d_1 t\Big\}\,.
\]
So (\ref{eqD.6}) is equivalent to $\mathbb P(\mathcal A) \le 2\exp(-t)$. Obviously $\mathbb P(\mathcal A \bigcap \mathcal E^c) \le \exp(-t)$. We bound $\mathbb P(\mathcal A \bigcap \mathcal E)$ as follows:
\begin{equation*}
\begin{aligned}
\mathbb P\Big(\mathcal A \bigcap \mathcal E\Big)  \le & \left. \mathbb P \Big( \by' (d_1 \bI_m - \bD^2) \by > \tr(d_1 \bI_m - \bD^2) \right. \\
& \left. + C_{\psi} \sqrt{t \sum\nolimits_{j \le m} d_j^2} + C_{\psi} d_1 t  - 2d_1\sqrt{m M_2 t} - 2 d_1 M_1^{-1} t\right. \Big) \\
\le & \left. \mathbb P \Big( \by' (d_1 \bI_m - \bD^2) \by > \tr(d_1 \bI_m - \bD^2) \right. \\
& \left. + 2\sqrt{t \sum\nolimits_{j \le m} (d_1 - d_j)^2} + 2(d_1 - d_m) t  \right. \Big) \le \exp(-t)\,,
\end{aligned}
\end{equation*}
where the first inequality is a summation of  the inequalities defined in events $\mathcal A$ and $\mathcal E$; the second inequality is due to the fact $\psi \sqrt{\sum\nolimits_{j \le m} d_j^2} \ge d_1\sqrt{m} \ge \sqrt{\sum\nolimits_{j \le m} (d_1 - d_j)^2}$ where $\psi = d_1/d_m $ and the last inequality is by (\ref{eqD.5}). Thus we have proved $\mathbb P(\mathcal A) \le 2\exp(-t)$.
\end{proof}

\bibliographystyle{ims}
\bibliography{Reference}

\end{document}